\documentclass[11pt]{article}
\usepackage{lmodern}
\usepackage{lmodern}
\usepackage[T1]{fontenc}
\usepackage[utf8]{inputenc}
\usepackage{geometry}
\usepackage{color}
\usepackage{mathtools}
\usepackage{amsmath}
\usepackage{amsthm}
\usepackage{amssymb}
\usepackage{graphicx}
\usepackage[bookmarks=false,
 breaklinks=false,pdfborder={0 0 1},backref=section,colorlinks=true]
 {hyperref}
\hypersetup{
 linkcolor=blue,citecolor=blue,urlcolor=blue}

\makeatletter
\theoremstyle{plain}
\newtheorem{thm}{\protect\theoremname}
\theoremstyle{definition}
\newtheorem{defn}[thm]{\protect\definitionname}
\theoremstyle{remark}
\newtheorem{rem}[thm]{\protect\remarkname}
\theoremstyle{plain}
\newtheorem{cor}[thm]{\protect\corollaryname}
\theoremstyle{definition}
\newtheorem{example}[thm]{\protect\examplename}
\theoremstyle{plain}
\newtheorem{prop}[thm]{\protect\propositionname}

\@ifundefined{date}{}{\date{}}

\usepackage{amsthm}\usepackage{booktabs}
\usepackage{array}
\usepackage{enumitem}

\newtheorem{theorem}{Theorem}[section]

\newtheorem{assumption}[theorem]{Assumption}

\newcommand{\Risk}{\mathrm{Risk}}

\title{The Epistemic Risk of Risk:\\ A Modal Framework for Quantitative Risk Management}
\author{Hirbod Assa\\University College Dublin}

\makeatother

\providecommand{\corollaryname}{Corollary}
\providecommand{\definitionname}{Definition}
\providecommand{\examplename}{Example}
\providecommand{\propositionname}{Proposition}
\providecommand{\remarkname}{Remark}
\providecommand{\theoremname}{Theorem}

\begin{document}
\maketitle
\begin{abstract}
Risk governance is a multilayer practice. Standard QRM focuses on
identifying, modelling, and measuring adverse states of the world.
Yet institutional risk governance also asks when an institution is
entitled to rely on a risk claim. This paper introduces modal epistemic
machinery to represent that entitlement. For a risk proposition $p$,
$Kp$ denotes assurance-grade knowledge, suitable for certification,
audit reliance, board sign-off, or regulatory reporting. By contrast,
$Bp$ denotes working commitment: a disciplined action-guiding stance
available under incomplete assurance. Following standard modal notation,
$K$ stands for knowledge and $B$ for belief; in this paper they
are interpreted institutionally.

The paper proceeds in five steps. First, it develops an institutional
modal language for QRM, interpreting $K$ as assurance-grade endorsement
and $B$ as working commitment. It studies the governance meaning
of factivity, consistency, closure, introspection, and reachability,
and compares the weaker packages used in the paper with standard S5
and KD45 systems. Probability remains indispensable in QRM, but probability
alone does not represent the institutional standing of a risk claim.

Second, the paper introduces epistemic-risk diagnostics. The central
diagnostics are 
\[
p\wedge\neg Kp\qquad\text{and}\qquad p\wedge\neg Bp.
\]
They identify cases in which a risk is real but lacks the relevant
institutional stance. Thus, alongside modelling object-level risks,
QRM should also model epistemic and governance conditions such as
evidential incompleteness, model risk, implementation uncertainty,
validation gaps, and failures of escalation.

Third, the framework is developed in both crisp and fuzzy modal semantics.
The fuzzy setting represents propositions by degrees and evidential
accessibility by fuzzy relations. This captures graded support, live
possibility, non-exclusion, structural uncertainty, epistemic inconsistency,
and hesitation.

Fourth, the paper studies two natural governance principles. The Risk
Management Principle says that if $p$ is a real risk, then the absence
of the relevant stance, 
\[
p\wedge\neg Mp,\qquad M\in\{K,B\},
\]
is itself risk-relevant. The Risk Reach Principle says that real and
decision-relevant risks should be reachable by the appropriate institutional
stance. The limitation is Moorean: if diagnostics such as $p\wedge\neg Kp$
and $p\wedge\neg Bp$ are treated as ordinary targets of the same
stance whose absence they record, they generate Moorean and Fitch-style
collapse pressure. In the fuzzy setting, this pressure is localised:
under factive knowledge-like standards, true-but-unknown risk is bounded
by reachable structural uncertainty or conflict; under non-factive
belief-like standards, true-but-unbelieved risk is bounded by reachable
epistemic inconsistency.

Finally the paper's proposal to address the above mentioned paradox
is architectural. Object-level risk claims should be separated from
meta-level epistemic diagnostics. The latter should be governed through
an audit layer that records and controls epistemic gaps, rather than
being forced into the same endorsement or commitment register whose
absence they diagnose. This preserves justified action and precaution
without collapsing into an ideal of institutional omniscience. \footnote{The project arose from a critical engagement with the treatment of
quantitative risk in \cite{EmbrechtsHofertChavezDemoulin2024RiskRevealed},
and develops the modal-epistemic perspective that the present paper
argues is missing from that account. }
\end{abstract}
\noindent\textbf{Keywords.} Epistemic risk; quantitative risk management;
modal epistemic logic; institutional knowledge; working commitment;
Moorean diagnostics; Fitch-style collapse; fuzzy modal semantics;
model risk; probability; meta-level controls. 

\section{Introduction}

Risk management is often described as decision-making under uncertainty.
This is correct, but incomplete. Risk governance also concerns institutional
reliance: which risk claims have sufficient authority to guide action.

This second question cannot be reduced to the first. A proposition
may be true without being institutionally recognised. Conversely,
an institution may endorse a true claim in a fragile way: the conclusion
is correct in the actual case, but nearby changes in data, model class,
scenario design, implementation, or evidential access would have produced
false reassurance.

This paper introduces modal epistemic machinery to represent this
second-order layer of risk governance.\footnote{The aim is not to replace probabilistic QRM. Probability remains central
for modelling likelihood, tail loss, exposure, capital impact, and
expected loss. The point is that probability alone does not represent
the institutional standing of a risk claim. A claim may be probable
but not robust, live but not endorsed, endorsed but unsafe, or excluded
under one institutional standard while remaining live under another.} For a risk proposition $p$, the formula $Kp$ means that the institution
has assurance-grade knowledge of $p$. This is the status required
for certification, board sign-off, regulatory reporting, audit reliance,
or formal institutional use. By contrast, $Bp$ means that the institution
has a working commitment to $p$: a disciplined action-guiding stance
available when assurance is incomplete. Such a commitment may justify
monitoring, stress testing, precautionary action, or temporary model
restrictions.

The distinction between $K$ and $B$ is central. Risk management
cannot always wait for assurance-grade knowledge. Severe risks may
require action while evidence is incomplete, models are contested,
implementation is uncertain, or validation is ongoing. Treating every
action-guiding stance as knowledge would over-certify the institution's
position; refusing to act before knowledge is available would understate
the demands of precaution.

The main claim is that QRM faces an epistemic risk of risk.\footnote{In modal logic, matters concerning belief are often called doxastic,
while matters concerning knowledge are called epistemic. In this paper,
epistemic is used in a broader institutional sense to cover both knowledge-like
and belief-like stances, unless a distinction is explicitly needed.} The epistemic risk of risk is the risk that an institution fails
to make a real risk visible, credible, or actionable through its evidential
processes, implementation standards, modelling choices, validation
procedures, or escalation channels. A real risk may fail to enter
either the assurance-grade procedure or the working-commitment procedure.
In that case, the institution is exposed both to the object-level
risk and to a failure in its relation to that risk.

The modal language makes these distinctions explicit: 
\[
p,\qquad Kp,\qquad Bp,\qquad\bar{K}p,\qquad\bar{B}p,\qquad\Diamond_{K}p,\qquad\Diamond_{B}p,\qquad p\wedge\neg Kp,\qquad p\wedge\neg Bp.
\]
Here $p$ is the object-level risk claim, such as ``the model underestimates
tail loss,'' ``the banking network is vulnerable to contagion,''
or ``the flood zone is critical.'' The formula $Kp$ says that $p$
has assurance-grade endorsement, while $Bp$ says that $p$ has working-commitment
status. The formulas 
\[
\bar{M}p:=\neg M\neg p,\qquad M\in\{K,B\},
\]
express non-exclusion: under standard $M$, the institution has not
ruled out $p$. The possibility formulas $\Diamond_{M}p$ express
live reachability: $p$ remains live somewhere in the relevant $M$-evidence
field. Thus $\Diamond_{K}p$ says that $p$ remains reachable under
the assurance standard, while $\Diamond_{B}p$ says that $p$ remains
live under the working-commitment standard. In many crisp settings,
$\Diamond_{M}p$ and $\bar{M}p$ coincide; in fuzzy settings, they
may differ depending on the chosen modal operations. 

The Moorean diagnostics record epistemic gaps: 
\[
p\wedge\neg Kp\qquad\text{and}\qquad p\wedge\neg Bp.
\]
The first says that the risk is real but lacks assurance-grade endorsement.
The second says that the risk is real but has not been adopted as
a working commitment. These formulas express audit-relevant governance
failures: a risk is present, but the institution lacks the relevant
stance toward it.

The paper contributes to the literature on quantitative risk in several
ways. First, it develops an institutional modal language for QRM.
In this language, $Kp$ is interpreted as assurance-grade endorsement
of a risk claim, while $Bp$ is interpreted as working commitment
under incomplete assurance. The framework also clarifies the governance
meaning of standard modal assumptions such as factivity, consistency,
closure, introspection, and reachability, and compares the weaker
packages used here with standard S5 and KD45 systems.

Second, it defines epistemic-risk diagnostics such as 
\[
p\wedge\neg Kp\qquad\text{and}\qquad p\wedge\neg Bp.
\]
These diagnostics identify cases in which a risk is real but lacks
the relevant institutional stance. They also show why QRM should model
not only hazards and losses, but also epistemic and governance conditions
such as evidential incompleteness, model uncertainty, implementation
uncertainty, validation gaps, and failures of recognition or escalation.

Third, the framework is developed in both crisp and fuzzy modal semantics.
The crisp setting recovers standard Kripke-style necessity and possibility.
The fuzzy setting allows propositions and evidence relations to hold
by degree, capturing graded support, live possibility, non-exclusion,
structural uncertainty, epistemic inconsistency, and hesitation.

Fourth, the paper introduces two natural principles of risk governance.
The Risk Management Principle says that real risks without the relevant
stance are themselves risk-relevant: 
\[
p\in\mathrm{Risk}(S)\quad\Rightarrow\quad p\wedge\neg Mp\in\mathrm{Risk}(S),\qquad M\in\{K,B\}.
\]
The Risk Reach Principle says that real and decision-relevant risks
should be reachable by the appropriate institutional stance: 
\[
p\in\mathrm{Risk}(S)\quad\Rightarrow\quad p\to\Diamond_{M}Mp.
\]
Here $\mathrm{Risk}(S)$ is the set of risk propositions relevant
to subject $S$. Together, these principles express the ambition that
real risks should not remain permanently outside the reach of evidential
processes, implementation, modelling, validation, endorsement, or
disciplined commitment. Their unrestricted combination, however, creates
Moorean and Fitch-style collapse pressure.

Fifth, motivated by these collapse results, the paper proposes a typed
architecture. Object-level risk claims are distinguished from meta-level
epistemic diagnostics. Object-level claims concern losses and adverse
states of the world. Meta-level diagnostics concern the institution's
relation to those claims: whether they are known, believed, live,
excluded, validated, escalated, owned, or left unresolved. This architecture
preserves the ambition that real risks should be reachable, while
avoiding the implausible conclusion that every real risk is already
known, believed, or institutionally endorsed.

The formal analysis is then connected back to applications. Model
risk illustrates the difference between a model output and assurance-grade
robustness under admissible model variation. Banking-network contagion
illustrates why working commitment need not be factive: an institution
may act on a severe stress scenario before the adverse state is actual.
Flood governance illustrates the geometry of endorsement, live possibility,
Moorean fragility, and hesitation. The paper also compares modal status
with probabilistic alternatives. Probability remains indispensable,
but it does not by itself encode the evidential geometry of endorsement,
non-exclusion, robustness, or institutional standing.

The paper is organised as follows. Section~\ref{sec:related} situates
the argument in the relevant literatures. Section~\ref{sec:setup}
develops the institutional epistemic setup. Section~\ref{sec:moorean}
states the Moorean and Fitch-style collapse pressure in crisp and
fuzzy form. Section~\ref{sec:examples} presents the examples motivated
by real applications. Section~\ref{sec:Probability-and-modal} compares
modal status with probabilistic alternatives. Section~\ref{sec:meta}
develops the proposed meta-level governance architecture. Section~\ref{sec:conclusion}
concludes. Proofs are collected in Appendix~\ref{app:proofs}. 

\section{Related Literature}

\label{sec:related}

This paper brings modal epistemic tools into QRM. Its background lies
in five literatures: anti-luck and anti-risk epistemology; Moorean
and Fitch-style epistemic logic; interactive epistemology and game-theoretic
knowledge; fuzzy modal methods for graded support and approximation;
and probabilistic or decision-theoretic treatments of epistemic risk.

\subsection{Epistemic risk, safety, and institutional entitlement}

Anti-luck epistemology holds that a true belief fails to amount to
knowledge when its truth is too lucky \cite{Pritchard2005EpistemicLuck,Pritchard2007AntiLuck,Pritchard2012AntiLuckVirtue}.
Safety accounts sharpen this thought: knowledge requires that the
relevant cognitive success could not easily have been a failure \cite{Sainsbury1997Easy,Sosa1999HowToDefeat,Williamson2000Knowledge}.
Pritchard's anti-risk epistemology gives the point its most direct
form: the relevant defect is exposure to nearby error \cite{Pritchard2015Risk,Pritchard2016EpistemicRisk}.

Institutional risk governance has direct analogues of these failures.
Pritchard's distinction between probability and modal closeness is
important here: a low-probability event may be modally close if only
a small perturbation is required for it to occur, while a non-negligible
probability may correspond to a structurally remote event \cite{Pritchard2016EpistemicRisk,Hawthorne2003Knowledge}.
Classic cases of epistemic luck similarly have institutional counterparts:
a risk process may reach the right conclusion by an unsafe route,
or operate in an environment where nearby cases would be misclassified
\cite{Chisholm1977Theory,Goldman1976Discrimination}.

\subsection{Moorean phenomena and knowability}

Moorean sentences such as ``$p$, but I do not believe $p$'' and
``$p$, but you do not know $p$'' have long been treated as pragmatically
or epistemically unstable \cite{Moore1942Reply,Hintikka1962Knowledge}.
In dynamic epistemic logic, they are analysed through learning and
public announcement: Moorean formulas may be unsuccessful or self-refuting
because learning them removes the ignorance or non-belief that made
them true \cite{Plaza1989Public,Gerbrandy1999Bisimulations,vanDitmarschHoekKooi2008Dynamic}.

Holliday and Icard show that Moorean phenomena are structural in important
epistemic and doxastic systems: in S5-like knowledge settings they
generate self-refutation, and in KD45-like belief settings they generate
unsuccessfulness \cite{HollidayIcard2010Moorean}. Fitch's knowability
paradox shows that unrestricted principles linking truth to knowability
can collapse into the conclusion that all truths are known \cite{Fitch1963Logical}.
The present paper imports this pressure into QRM. The formulas $p\wedge\neg Kp$
and $p\wedge\neg Bp$ are useful audit diagnostics, but they become
unstable if the same institution is required to know or believe them
as ordinary object-level outputs.

\subsection{Interactive epistemology and institutional fragmentation}

Interactive epistemology and game theory provide tools for modelling
knowledge, belief, common knowledge, and hierarchies of belief. Battigalli
and Bonanno give a state-space treatment of belief and knowledge in
which possibility correspondences encode what agents consider possible
and what they believe \cite{BattigalliBonanno1999Recent}. In relational
semantics, seriality corresponds to consistency of belief, transitivity
to positive introspection, and Euclideanness to negative introspection;
adding reflexivity gives the S5-style structure associated with knowledge
\cite{Chellas1980Modal,FaginHalpernMosesVardi1995Reasoning,Geanakoplos1994Common}.

This matters because institutions are fragmented epistemic systems.
 Many governance failures arise because one unit assumes that another
knows, owns, or has escalated a risk. The literature on common knowledge
and game-theoretic epistemology is therefore directly relevant \cite{Aumann1976Agreeing,Aumann1987Correlated,Aumann1995Backward,AumannBrandenburger1995Epistemic,Harsanyi1967Games,BrandenburgerDekel1993Hierarchies}.
Epistemic assumptions also underlie rationalizability, correlated
equilibrium, forward induction, and backward induction \cite{Bernheim1984Rationalizable,Pearce1984Rationalizable,BrandenburgerDekel1987Rationalizability,DekelGul1997Rationality,BattigalliSiniscalchi1998Hierarchies}.
The present paper uses these ideas institutionally rather than strategically:
the question is what a risk institution is entitled to know, believe,
treat as live, or regard as ruled out.

\subsection{Fuzzy modal methods}

Classical modal logic uses crisp accessibility relations and crisp
propositions. QRM often requires a graded version. Evidence may be
partial, vague, asymmetric, or uncertain; a risk may be strongly live,
weakly excluded, partially supported, or robust only under restricted
assumptions. Fuzzy modal logic generalises Kripke semantics by replacing
crisp relations with fuzzy relations and evaluating necessity, possibility,
sufficiency, and dual sufficiency by degree \cite{Zadeh1965FuzzySets,Zadeh1972Linguistic,Hajek1998Metamathematics,Radzikowska2017FuzzyModal}.

The construction also connects with rough-set lower and upper approximations
\cite{Pawlak1982RoughSets,Pawlak1991RoughSets}. The lower approximation
corresponds to what the institution can positively endorse; the upper
approximation corresponds to what it cannot responsibly exclude. The
gap between them is a hesitation region, related to intuitionistic
fuzzy sets and distance-based accounts of hesitation \cite{Atanassov1986Intuitionistic,Atanassov1999Intuitionistic,SzmidtKacprzyk2000Distances}.
The algebraic flexibility of t-norms and fuzzy implications allows
different standards of support to be modelled \cite{KlementMesiarPap2000Triangular,BaczynskiJayaram2008Fuzzy}.
Linguistic hedges such as material, plausible, remote, reasonably
foreseeable, not ruled out, and sufficiently controlled behave like
modal governance predicates rather than merely like probabilities
\cite{DeCockRadzikowskaKerre2002FuzzyRough,DeCockKerre2004Fuzzy}.

\subsection{Probabilistic and decision-theoretic approaches}

A natural alternative is to treat epistemic risk probabilistically.
Babic develops an alethic theory of epistemic risk in which the riskiness
of a credence function measures sensitivity to graded error and connects
to expected inaccuracy and information entropy \cite{Babic2019TheoryEpistemicRisk}.
Pettigrew develops a decision-theoretic account in which rational
agents may permissibly differ in their attitudes to epistemic risk,
encoded by epistemic decision rules \cite{Pettigrew2022EpistemicRisk}.

The present paper studies a different object. Its target is not credence
alone, nor the decision rule by which an individual selects rational
credences, but institutional epistemic status under assurance, working
commitment, non-exclusion, and audit. These statuses require modal
structure: factivity for knowledge-like endorsement, consistency for
belief-like commitment, closure behaviour, introspection, and reachability
across admissible evidence states. A probabilistic threshold account,
\[
Bp\quad\text{iff}\quad\mathbb{P}(p)\ge t,
\]
does not generally supply these structures and is vulnerable to lottery-style
and preface-style failures \cite{Kyburg1961Probability,Makinson1965Preface,Foley1992WorkingWithoutNet,Christensen2004PuttingLogic}.
Probability remains indispensable, but it does not by itself represent
the institutional accessibility and endorsement structure studied
here.

\section{The Epistemic Setup for QRM}

\label{sec:setup}

This section introduces the modal framework. The aim is to distinguish
object-level risk claims from the institutional status of those claims.

\subsection{Graded propositions and evidence relations}

Let $W\neq\varnothing$ be a set of possible worlds, and let 
\[
\mathcal{L}\subseteq[0,1]^{W}
\]
be a collection of graded propositions. Each $p\in\mathcal{L}$ is
a function 
\[
p:W\to[0,1],
\]
where $p(w)$ is the degree to which $p$ holds at $w$. The crisp
case is recovered when $p(w)\in\{0,1\}$. Thus a risk claim may hold
fully, fail fully, or hold to an intermediate degree.

Equip $\mathcal{L}$ with the pointwise order: 
\[
p\le q\quad:\Longleftrightarrow\quad\forall w\in W,\;p(w)\le q(w).
\]
Thus $p\le q$ means that $p$ is nowhere stronger than $q$. Assume
that $\mathcal{L}$ is closed under conjunction $\wedge$ and negation
$\neg$. In the standard case, 
\[
(p\wedge q)(w)=\min\{p(w),q(w)\}.
\]
Thus conjunction records the degree to which both claims hold at the
same world.
\begin{defn}
[Structural uncertainty] For $p\in\mathcal{L}$, define 
\[
U(p):=p\wedge\neg p.
\]
Define global structural uncertainty by 
\[
U(w):=\sup_{p\in\mathcal{L}}U(p)(w).
\]
\end{defn}

This measures the degree to which a proposition overlaps with its
negation; globally, it records the largest such overlap at $w$.

Fix an epistemic standard 
\[
M:\mathcal{L}\to\mathcal{L}.
\]
Let 
\[
\gamma_{M}:W\times W\to[0,1]
\]
be a fuzzy evidence relation associated with $M$. The value $\gamma_{M}(w,v)$
is the degree to which $v$ is evidentially relevant from $w$ under
standard $M$. Thus $\gamma_{M}$ specifies which alternatives matter,
and to what degree, when applying the institutional standard $M$.

Let $\otimes$ be a monotone combination rule, typically a t-norm,
and let $\Rightarrow$ be an implication operation. Define 
\[
(Mp)(w):=\inf_{v\in W}\bigl(\gamma_{M}(w,v)\Rightarrow p(v)\bigr),
\]
and 
\[
(\Diamond_{M}p)(w):=\sup_{v\in W}\bigl(\gamma_{M}(w,v)\otimes p(v)\bigr).
\]
Thus $Mp$ is the degree of support across the evidence field, while
$\Diamond_{M}p$ is the degree to which $p$ remains live. In governance
terms, $Mp$ represents endorsement or support, while $\Diamond_{M}p$
represents live possibility or non-exclusion.

Define the dual non-exclusion operator by 
\[
\overline{M}p:=\neg M\neg p.
\]
Thus $\overline{M}p$ says that $p$ has not been ruled out under
standard $M$.

The hesitation margin is 
\[
H_{M}(p):=\max\{0,\overline{M}p-Mp\}.
\]
Thus $H_{M}(p)$ measures the gap between what the institution cannot
rule out and what it can positively endorse. 
\begin{rem}
[Possibility and dual non-exclusion] Both $\Diamond_{M}p$ and $\overline{M}p:=\neg M\neg p$
express non-exclusion. In many standard modal settings they coincide.
In particular, in crisp Kripke semantics with Boolean negation, 
\[
\Diamond_{M}p\equiv\neg M\neg p.
\]
Thus $p$ is possible under $M$ exactly when $p$ is not ruled out
by $M$.

In the fuzzy setting, the two notions may also coincide when the operations
$\otimes$, $\Rightarrow$, and $\neg$ form a suitable dual modal
pair. For example, under standard choices of fuzzy negation and residuated
implication, the possibility operator can be defined as the dual of
necessity: 
\[
\Diamond_{M}p:=\neg M\neg p.
\]
However, the paper keeps the notation $\Diamond_{M}p$ and $\overline{M}p$
separate because some fuzzy packages define possibility directly by
\[
(\Diamond_{M}p)(w)=\sup_{v\in W}\bigl(\gamma_{M}(w,v)\otimes p(v)\bigr),
\]
while defining non-exclusion by 
\[
\overline{M}p:=\neg M\neg p.
\]
In such cases the two operators may differ. The distinction is useful:
$\Diamond_{M}p$ records live evidential reachability, while $\overline{M}p$
records failure to rule out $p$ through the dual of the support operator. 
\end{rem}

\subsection{Special cases}

\subsubsection*{Crisp evidence relations}

Let $W\neq\varnothing$ be a set of worlds. A crisp proposition is
a subset $T(p)\subseteq W$. Let $R_{M}\subseteq W\times W$ be an
accessibility relation, and write 
\[
\Gamma_{M}(w):=\{v\in W:wR_{M}v\}.
\]
Then 
\[
w\models Mp\Longleftrightarrow\Gamma_{M}(w)\subseteq T(p),
\]
and 
\[
w\models\Diamond_{M}p\Longleftrightarrow\Gamma_{M}(w)\cap T(p)\neq\varnothing.
\]
Thus $Mp$ means that $p$ holds throughout the evidence set, while
$\Diamond_{M}p$ means that $p$ holds somewhere in it.

Equivalently, if $p$ is represented by its indicator function $p:W\to\{0,1\}$,
then 
\[
(Mp)(w)=\bigwedge_{v\in\Gamma_{M}(w)}p(v),\qquad(\Diamond_{M}p)(w)=\bigvee_{v\in\Gamma_{M}(w)}p(v).
\]
Hence necessity is conjunction over the evidence set, while possibility
is disjunction over the evidence set.

This reduction is not specific to the G\"odel/min package. It follows
in the crisp case for any standard residuated fuzzy package satisfying
the boundary conditions 
\[
1\Rightarrow a=a,\qquad0\Rightarrow a=1,
\]
and 
\[
1\otimes a=a,\qquad0\otimes a=0.
\]
Indeed, if $\gamma_{M}(w,v)\in\{0,1\}$, then worlds outside $\Gamma_{M}(w)$
contribute $1$ to the infimum defining $Mp$, and worlds inside $\Gamma_{M}(w)$
contribute $p(v)$. Thus 
\[
(Mp)(w)=\inf_{v\in W}\bigl(\gamma_{M}(w,v)\Rightarrow p(v)\bigr)=\bigwedge_{v\in\Gamma_{M}(w)}p(v).
\]
Similarly, worlds outside $\Gamma_{M}(w)$ contribute $0$ to the
supremum defining $\Diamond_{M}p$, and worlds inside $\Gamma_{M}(w)$
contribute $p(v)$. Hence 
\[
(\Diamond_{M}p)(w)=\sup_{v\in W}\bigl(\gamma_{M}(w,v)\otimes p(v)\bigr)=\bigvee_{v\in\Gamma_{M}(w)}p(v).
\]
With Boolean negation, the dual non-exclusion operator also coincides
with possibility: 
\[
\overline{M}p:=\neg M\neg p=\Diamond_{M}p.
\]

\subsubsection*{Metric evidence relations}

Let $(W,d)$ be a metric space and fix $\beta>0$. Define 
\[
wRv\Longleftrightarrow d(w,v)\le\beta.
\]
Then 
\[
\Gamma(w)=B_{\beta}(w).
\]
Hence 
\[
w\models Mp\Longleftrightarrow p(v)=1\text{ for all }v\in B_{\beta}(w),
\]
and 
\[
w\models\Diamond_{M}p\Longleftrightarrow B_{\beta}(w)\cap T(p)\neq\varnothing.
\]
For $W=\mathbb{R}^{n}$, one may use 
\[
d_{a}(w,v)^{2}:=\sum_{i=1}^{n}a_{i}|x_{i}-y_{i}|^{2},\qquad a_{i}>0.
\]
The evidence set is then an ellipsoid around $w$.

\subsubsection*{Information cells and S5 knowledge}

If $R_{M}$ is an equivalence relation, then $\Gamma_{M}(w)$ is the
information cell of $w$. In this case, 
\[
w\models Kp\Longleftrightarrow\Gamma_{K}(w)\subseteq T(p).
\]
This is the standard S5 interpretation of knowledge \cite{Hintikka1962Knowledge,FaginHalpernMosesVardi1995Reasoning,MeyerVanDerHoek1995Epistemic,BlackburnDeRijkeVenema2001Modal}.

The representation is useful when risk information is partitioned
into reporting states, supervisory categories, or formally recognised
information sets.

\subsubsection*{Preorders and belief-like refinement}

A preorder can represent evidential refinement or informational strength.
A belief-like stance is then 
\[
w\models Bp\Longleftrightarrow\Gamma_{B}(w)\subseteq T(p).
\]
Unlike knowledge, this stance need not be factive. The actual world
need not belong to $\Gamma_{B}(w)$. This allows $B$ to represent
precautionary working commitments based on stress-relevant alternatives.

\subsection{Basic modal assumptions and their frame conditions}

\label{sec:basic-modal-assumptions}

This section records the modal assumptions used later and compares
them with the standard S5 and KD45 packages.
\begin{defn}
An operator $M:\mathcal{L}\to\mathcal{L}$ is \emph{monotone} if 
\[
p\le q\quad\Rightarrow\quad Mp\le Mq.
\]
The same definition applies to $\Diamond_{M}$.
\end{defn}

In the fuzzy semantics above, sufficient conditions for monotonicity
are straightforward. If the implication $\Rightarrow$ is increasing
in its second argument, then 
\[
p\le q\quad\Rightarrow\quad\gamma_{M}(w,v)\Rightarrow p(v)\le\gamma_{M}(w,v)\Rightarrow q(v)
\]
for all $w,v$. Taking infima gives 
\[
Mp\le Mq.
\]
Similarly, if $\otimes$ is increasing in its second argument, then
\[
p\le q\quad\Rightarrow\quad\gamma_{M}(w,v)\otimes p(v)\le\gamma_{M}(w,v)\otimes q(v),
\]
and taking suprema gives 
\[
\Diamond_{M}p\le\Diamond_{M}q.
\]

\begin{defn}
An operator $M$ is \emph{factive} if 
\[
Mp\le p
\]
for all $p\in\mathcal{L}$.
\end{defn}

In crisp notation, this is $Mp\to p$. Factivity is appropriate for
knowledge-like assurance: if the institution has assurance-grade endorsement
of $p$, then $p$ should be correct.

A sufficient condition for factivity is 
\[
\gamma_{M}(w,w)=1\quad\text{for all }w\in W,
\]
together with 
\[
1\Rightarrow a=a.
\]
Indeed, 
\[
(Mp)(w)=\inf_{v\in W}\bigl(\gamma_{M}(w,v)\Rightarrow p(v)\bigr)\le\gamma_{M}(w,w)\Rightarrow p(w)=1\Rightarrow p(w)=p(w).
\]
Hence $Mp\le p$.
\begin{defn}
An operator $M$ is \emph{positively introspective} if 
\[
Mp\le MMp
\]
for all $p\in\mathcal{L}$.
\end{defn}

Institutionally, this says that if a claim is in the relevant $M$-register,
then the fact that it is in that register is itself available to the
register. In crisp Kripke semantics, a sufficient condition for positive
introspection is transitivity of the accessibility relation: 
\[
wR_{M}u\ \text{and}\ uR_{M}v\quad\Rightarrow\quad wR_{M}v.
\]
In fuzzy semantics, the corresponding sufficient condition is fuzzy
transitivity: 
\[
\gamma_{M}(w,u)\otimes\gamma_{M}(u,v)\le\gamma_{M}(w,v)\qquad\text{for all }w,u,v\in W.
\]
Assuming $\Rightarrow$ is the residuum of $\otimes$, this condition
yields 
\[
Mp\le MMp.
\]
Thus positive introspection follows when the evidence relation is
transitive in the relevant crisp or fuzzy sense.

Define epistemic inconsistency by 
\[
I_{M}(p):=Mp\wedge M\neg p,
\]
and global epistemic inconsistency by 
\[
I_{M}(w):=\sup_{p\in\mathcal{L}}I_{M}(p)(w).
\]
If $M$ is factive, then 
\[
I_{M}(p)\le U(p).
\]
Indeed, factivity gives $Mp\le p$ and $M\neg p\le\neg p$. Hence
\[
I_{M}(p)=Mp\wedge M\neg p\le p\wedge\neg p=U(p).
\]

It is useful to compare these assumptions with the two standard modal
packages for knowledge and belief: S5 and KD45. In crisp Kripke semantics,
knowledge is often modelled by S5. This corresponds to an accessibility
relation that is reflexive, transitive, and Euclidean; equivalently,
an equivalence relation. The associated principles include 
\[
Kp\to p,
\]
\[
Kp\to KKp,
\]
and 
\[
\neg Kp\to K\neg Kp.
\]
These are factivity, positive introspection, and negative introspection.
S5 is a useful benchmark for assurance-grade knowledge. If $Kp$ represents
formal institutional certification, then factivity is natural. Positive
introspection is also plausible when endorsements are formally recorded.
Negative introspection is stronger: institutions may fail to know
that they do not know. For this reason, the paper does not require
full S5.

Belief is often modelled by KD45. In crisp Kripke semantics, KD45
corresponds to a serial, transitive, and Euclidean accessibility relation.
Its characteristic principles include consistency, 
\[
Bp\to\neg B\neg p,
\]
positive introspection, 
\[
Bp\to BBp,
\]
and negative introspection, 
\[
\neg Bp\to B\neg Bp.
\]
KD45 is a benchmark for idealised belief. In this paper, however,
$B$ is not private belief. The institution may act on a severe stress
scenario before that scenario is actual. Thus the paper uses a weaker
belief-like package than KD45.

The technical results rely on two weaker packages.
\begin{itemize}
\item \textbf{The }\textbf{\emph{factive package}}\textbf{.} This package
used for knowledge-like standards, assumes: 
\[
p\le q\Rightarrow Mp\le Mq,
\]
\[
p\le q\Rightarrow\Diamond_{M}p\le\Diamond_{M}q,
\]
and 
\[
Mp\le p.
\]
This package is weaker than S5. It does not require positive introspection,
negative introspection, or an equivalence relation.
\item \textbf{The }\textbf{\emph{non-factive package}}\textbf{.} This package
used for belief-like standards, assumes: 
\[
p\le q\Rightarrow Mp\le Mq,
\]
\[
p\le q\Rightarrow\Diamond_{M}p\le\Diamond_{M}q,
\]
and 
\[
Mp\le MMp.
\]
This package is weaker than KD45. It does not assume factivity, negative
introspection, or full consistency. Positive introspection is retained
because working commitments should be institutionally recordable:
if the organisation acts as if $p$, then the fact that it has adopted
this stance should itself be available to the relevant register.
\end{itemize}
The comparison is: {\scriptsize
\[
\begin{array}{c|c|c|c|c}
\text{Package} & \text{Use} & \text{Factivity} & \text{Positive introspection} & \text{Negative introspection}\\
\hline \text{S5} & \text{ideal knowledge} & \text{yes} & \text{yes} & \text{yes}\\
\text{KD45} & \text{ideal belief} & \text{no} & \text{yes} & \text{yes}\\
\text{Factive package} & \text{assurance-like }K & \text{yes} & \text{not required} & \text{not required}\\
\text{Non-factive package} & \text{working-commitment }B & \text{no} & \text{yes} & \text{not required}
\end{array}
\]
}{\scriptsize\par}

\subsection{Risk propositions and institutional stances}

Fix a subject $S$, such as a portfolio, firm, banking network, flood
system, infrastructure system, or model environment. Let $\Risk(S)$
be a family of risk-relevant propositions about $S$.

For $M\in\{K,B\}$, 
\[
Kp=\text{the institution has assurance-grade endorsement of }p,
\]
and 
\[
Bp=\text{the institution has a disciplined working commitment to }p.
\]
The possibility operator $\Diamond_{M}$ represents live possibility
or reachability: 
\[
\Diamond_{M}p=\text{\ensuremath{p} remains live under the \ensuremath{M}-evidence field.}
\]
Thus $Mp$ is positive support, while $\Diamond_{M}p$ is non-exclusion
or live possibility. The following table summarizes the institutional
reading of the risk a proposition and its meta level stances.

\[
\begin{array}{c|l}
\text{Formula} & \text{Institutional reading}\\
\hline p & \text{the object-level risk claim holds}\\
Kp & \text{the claim has assurance-grade endorsement}\\
Bp & \text{the claim has working-commitment status}\\
\Diamond_{M}p & \text{the claim remains live under standard }M\\
\overline{M}p & \text{the claim is not ruled out under standard }M\\
H_{M}(p) & \text{the hesitation gap between non-exclusion and endorsement}\\
p\wedge\neg Kp & \text{the risk is present but not assurance-endorsed}\\
p\wedge\neg Bp & \text{the risk is present but not action-committed}\\
\Diamond_{M}p\wedge\neg Mp & \text{the claim is live but not endorsed under standard }M\\
\overline{M}p\wedge\neg Mp & \text{the claim is not ruled out but not positively supported}
\end{array}
\]

\subsection{Epistemic-risk refinements}

\label{subsec:refinements}

An epistemic-risk refinement attaches an epistemic condition to an
object-level risk claim. It records not only whether a risk obtains,
but also whether the institution has the relevant stance toward it.

The central refinement is the Moorean diagnostic: 
\[
\mathcal{R}_{M}^{\mathrm{Moore}}(p):=p\wedge\neg Mp.
\]
This says that $p$ obtains, but the institution lacks the corresponding
$M$-stance toward $p$. If $M=K$, the risk is real but lacks assurance-grade
endorsement. If $M=B$, the risk is real but has not been adopted
as a working commitment. In governance terms, the risk may be absent
from the evidential route, implementation process, approved model,
validation report, or other procedure through which institutional
recognition is produced.

The anti-Moorean refinement is 
\[
\mathcal{R}_{M}^{\mathrm{anti}}(p):=p\wedge M\neg p.
\]
This says that $p$ obtains while the institution supports $\neg p$.
This is stronger than absence of endorsement. The institution is committed,
under standard $M$, to the opposite claim. For example, if $p$ says
that a model underestimates tail loss, then 
\[
p\wedge K\neg p
\]
says that the model does underestimate tail loss while the institution
has assurance-grade endorsement that it does not. This is an epistemic
conflict, not only an epistemic gap.

For two standards $E$ and $M$, the mixed unsupported refinement
is 
\[
\mathcal{R}_{E,M}^{\mathrm{unsup}}(p):=Ep\wedge\neg Mp.
\]
This says that $p$ is supported under one standard but not under
another. For example, an exploratory stress model may support a concern
that has not yet been validated by the official model-risk framework.
The risk then has support somewhere in the institution, but not under
the standard required for formal use.

The mixed conflict refinement is 
\[
\mathcal{R}_{E,M}^{\mathrm{conf}}(p):=Ep\wedge M\neg p.
\]
This says that one standard supports $p$, while another supports
$\neg p$. Such cases are common in fragmented institutions. A business
unit may treat a model as adequate while validation treats it as unreliable.
A flood authority may treat a scenario as remote while an engineering
team treats it as plausible. A supervisor may treat a contagion channel
as live while an internal capital model excludes it.

These refinements show why epistemic risk is not ordinary risk under
another name. A flood event, model defect, or contagion pathway is
an object-level risk. But being under-evidenced, unmodelled, poorly
implemented or unvalidated, that is being unsupported under the official
standard, or contradicted by another institutional standard, is a
meta-level governance condition attached to that risk.

The role of the refinements is diagnostic. They identify the institutional
status of a risk claim. This is the epistemic layer that ordinary
first-order risk modelling does not capture. 

\subsection{Two governance principles}

The modal setup gives rise to two natural governance principles. Each
is plausible on its own. The first treats epistemic failures attached
to real risks as risk-relevant. The second requires real and decision-relevant
risks to be reachable by responsible institutional processes. The
difficulty, developed later, is that the unrestricted combination
of the two creates Moorean collapse pressure.

\paragraph{The Risk Management Principle.}

Let $\mathcal{R}$ be an epistemic-risk refinement operator. The Risk
Management Principle says: 
\[
p\in\mathrm{Risk}(S)\quad\Rightarrow\quad\mathcal{R}(p)\in\mathrm{Risk}(S).
\]
For the Moorean refinement 
\[
\mathcal{R}_{M}^{\mathrm{Moore}}(p):=p\wedge\neg Mp,
\]
this becomes 
\[
p\in\mathrm{Risk}(S)\quad\Rightarrow\quad p\wedge\neg Mp\in\mathrm{Risk}(S).
\]

The principle says that if $p$ is a real risk, then the absence of
the relevant institutional stance toward $p$ is also risk-relevant.
A real risk that is not known, not adopted as a working commitment,
not evidenced, not implemented, not validated, or not modelled creates
a governance exposure. Epistemic absence is therefore not neutral:
failure to recognise or act on a real risk is itself part of the risk
landscape.

For example, if a model underestimates tail loss, then the absence
of this underestimation from the validation report is not only an
administrative omission. It is a risk-relevant condition. Similarly,
if a contagion channel remains live but has not entered the working-commitment
register, the institution is exposed both to the contagion risk and
to its own failure to treat that channel as actionable.

\paragraph{The Risk Reach Principle.}

The second principle is the Risk Reach Principle: 
\[
p\in\mathrm{Risk}(S)\quad\Rightarrow\quad p\le\Diamond_{M}Mp.
\]
In crisp notation, 
\[
p\in\mathrm{Risk}(S)\quad\Rightarrow\quad p\to\Diamond_{M}Mp.
\]

This says that if $p$ is a real and decision-relevant risk, then
the relevant institutional stance toward $p$ should be reachable.
The principle does not say that the institution already has $Mp$.
It says that the governance architecture should make $Mp$ attainable
through an admissible evidential route, implementation route, validation
route, modelling route, or escalation route.

Thus the principle is a reachability requirement on the risk process
itself. A risk framework should not be designed so that real risks
remain permanently outside the scope of evidence, implementation,
modelling, validation, or working commitment.

The two principles are independently natural. The Risk Management
Principle says that unrecognised real risks matter. The Risk Reach
Principle says that real risks should be capable of entering responsible
institutional recognition. Their unrestricted combination, however,
generates the Moorean limitation developed in Section~\ref{sec:moorean}. 

\section{The Moorean Limitation: Collapse Pressure}

\label{sec:moorean}

The preceding sections introduced diagnostics such as 
\[
p\wedge\neg Kp\qquad\text{and}\qquad p\wedge\neg Bp.
\]
These formulas are useful because they identify real risks that lack
the relevant institutional stance. The first says that a risk is present
but not assurance-endorsed. The second says that a risk is present
but not adopted as a working commitment.

The problem is that these diagnostics cannot be treated as ordinary
object-level targets of the same stance whose absence they record.
If the institution is required to know $p\wedge\neg Kp$, then it
is pushed toward knowing $p$, which conflicts with the $\neg Kp$
part. Similarly, if the institution is required to believe $p\wedge\neg Bp$,
then the working-commitment register becomes unstable under standard
introspective discipline.

This section states the formal pressure. Section~\ref{sec:meta}
then gives the architectural response: such diagnostics should be
governed at a meta level rather than forced back into the object-level
$K$- or $B$-register.

\subsection{Standing assumptions}

We use the following background assumptions.

\begin{assumption}[Meet] The operation $\wedge$ is a meet: 
\[
p\wedge q\le p,\qquad p\wedge q\le q,
\]
and 
\[
r\le p,\ r\le q\quad\Rightarrow\quad r\le p\wedge q.
\]
\end{assumption}

\begin{assumption}[Bottom preservation] 
\[
\Diamond_{M}0=0.
\]
\end{assumption}

\begin{assumption}[Conjunction separation] 
\[
p\wedge\neg q\le0\quad\Rightarrow\quad p\le q.
\]
\end{assumption}

These assumptions hold in the crisp Boolean case and in standard fuzzy
settings with minimum conjunction and standard negation.

Let $\mathcal{R}$ be a factive epistemic-risk refinement, meaning
\[
\mathcal{R}(p)\le p.
\]
Thus $\mathcal{R}(p)$ does not introduce a different object-level
risk. It refines $p$ by adding an epistemic condition, such as absence
of endorsement or conflict between standards.

\subsection{Factive pressure}

First consider the factive package from Section~\ref{sec:basic-modal-assumptions}.
Thus $M$ is monotone, $\Diamond_{M}$ is monotone, and 
\[
Mp\le p.
\]
This is the knowledge-like case. It applies to assurance-grade standards.
\begin{thm}
\label{thm:factive-pressure} Suppose RMP holds for $\mathcal{R}$,
and RRP holds for $M$. Then, for every $p\in\Risk(S)$, 
\[
\mathcal{R}(p)\le\Diamond_{M}\bigl(Mp\wedge\mathcal{R}(p)\bigr).
\]
\end{thm}

The theorem says that once an epistemic refinement of $p$ is treated
as a risk and is subjected to the same reachability principle, it
is pushed toward a reachable state in which both the refinement and
$Mp$ hold. This is the abstract form of the Moorean pressure.

For the Moorean refinement 
\[
\mathcal{R}_{M}^{\mathrm{Moore}}(p):=p\wedge\neg Mp,
\]
we obtain the following.
\begin{cor}
\label{cor:factive-moore} 
\[
p\wedge\neg Mp\le\Diamond_{M}(Mp\wedge\neg Mp)\le\Diamond_{M}U.
\]
\end{cor}

Thus true-but-unknown risk can persist only to the extent that structural
uncertainty is reachable. In the knowledge-like case, the Moorean
region is bounded by uncertainty or conflict in the evidential structure.

If structural uncertainty vanishes, the pressure becomes stronger.
\begin{cor}
\label{cor:factive-collapse} If $U=0$, then for every $p\in\Risk(S)$
satisfying the standing assumptions
\[
p=Mp.
\]
In the crisp case, this is the collapse 
\[
p\leftrightarrow Mp.
\]
\end{cor}

This is the factive version of the collapse. It does not say that
real institutions know every real risk. It says that an untyped formal
framework with factivity, reachability, and no structural uncertainty
cannot represent the space of true but unendorsed risk.

There is also a conflict version. Define 
\[
\mathcal{R}_{M}^{\mathrm{anti}}(p):=p\wedge M\neg p.
\]
This says that $p$ obtains while the institution supports $\neg p$.
\begin{cor}
\label{cor:factive-conflict} 
\[
p\wedge M\neg p\le\Diamond_{M}I_{M}(p)\le\Diamond_{M}I_{M}.
\]
\end{cor}

Thus factive conflict can persist only where $M$-inconsistency is
reachable. In governance terms, a certified false reassurance is possible
only where the assurance architecture itself permits uncertainty or
conflict.

\subsection{Non-factive pressure}

Now consider the non-factive package from Section~\ref{sec:basic-modal-assumptions}.
Thus $M$ is monotone, $\Diamond_{M}$ is monotone, and 
\[
Mp\le MMp.
\]
This is the belief-like case.
\begin{thm}
\label{thm:belief-internal} For every $p\in\mathcal{L}$, 
\[
M(p\wedge\neg Mp)\le I_{M}(Mp)\le I_{M}.
\]
\end{thm}

This theorem gives the internal non-factive pressure. If the institution
adopts the Moorean diagnostic $p\wedge\neg Mp$ as an $M$-commitment,
then the $M$-register becomes inconsistent about $Mp$. Thus belief
in a Moorean diagnostic is bounded by epistemic inconsistency.

Combining this with RMP and RRP gives the non-factive reach result.
\begin{thm}
\label{thm:belief-reach} For every $p\in\Risk(S)$, 
\[
p\wedge\neg Mp\le\Diamond_{M}I_{M}(Mp)\le\Diamond_{M}I_{M}.
\]
\end{thm}

Thus true-but-unbelieved risk can persist only to the extent that
inconsistency in the working-commitment register is reachable. The
mechanism differs from the factive case. For knowledge-like standards,
the pressure comes from correctness. For belief-like standards, it
comes from positive introspection and the discipline of the commitment
register.

If the register is fully coherent, the pressure again collapses the
Moorean region.
\begin{cor}
\label{cor:belief-collapse} If $I_{M}=0$, then 
\[
p\wedge\neg Mp\le0.
\]
If conjunction separation also holds, then 
\[
p\le Mp.
\]
In the crisp case, this is the collapse 
\[
p\to Mp.
\]
\end{cor}

Again, the conclusion is not descriptive. It is not saying that institutions
actually believe every real risk. It is a warning about the formal
architecture. If the framework treats true-but-uncommitted risk as
an ordinary object-level target of the same commitment operator, and
if the commitment register is idealised as coherent, then the space
of true but uncommitted risk disappears.

\section{Examples motivated by real applications}

\label{sec:examples}

Throughout this section, use the G\"odel/min package: 
\[
a\otimes b:=\min\{a,b\},\qquad a\Rightarrow b:=\begin{cases}
1, & a\le b,\\
b, & a>b.
\end{cases}
\]
Then 
\[
(\Diamond_{M}p)(w)=\sup_{v\in W}\min\{\gamma_{M}(w,v),p(v)\},
\]
and 
\[
(Mp)(w)=\inf_{v\in W}\bigl(\gamma_{M}(w,v)\Rightarrow p(v)\bigr).
\]
Thus $Mp$ represents support or endorsement across the evidence field,
while $\Diamond_{M}p$ represents live reachability. The dual operator
\[
\overline{M}p:=\neg M\neg p
\]
represents non-exclusion. In crisp Boolean cases, $\Diamond_{M}p$
and $\overline{M}p$ coincide. In fuzzy settings, they may differ
depending on the chosen operations. The hesitation margin 
\[
H_{M}(p):=\max\{0,\overline{M}p-Mp\}
\]
measures the gap between what is not ruled out and what is positively
endorsed.

\subsection{The two-world case}

\label{sec:two-world-case}

The two-world case is the smallest setting in which the main distinctions
of the paper appear. Let 
\[
W=\{w_{0},w_{1}\},
\]
where $w_{0}$ is the actual state and $w_{1}$ is a relevant alternative,
challenger state, or stress state. Write 
\[
p_{0}:=p(w_{0}),\qquad p_{1}:=p(w_{1}).
\]
In the crisp case, $p_{i}\in\{0,1\}$. In the fuzzy case, $p_{i}\in[0,1]$.
Thus $p_{0}$ measures the degree to which the risk obtains at the
actual state, while $p_{1}$ measures the degree to which it obtains
at the alternative state.

By the crisp reduction from Section~\ref{sec:setup}, at $w_{0}$
we have 
\[
Mp(w_{0})=\bigwedge_{v\in\Gamma_{M}(w_{0})}p(v),\qquad\Diamond_{M}p(w_{0})=\bigvee_{v\in\Gamma_{M}(w_{0})}p(v).
\]
The two-world table is therefore obtained by evaluating these conjunctions
and disjunctions for the four possible evidence sets. 

\subsubsection*{Crisp evidence sets}

First consider a crisp evidence relation. At the actual world $w_{0}$,
the evidence set is 
\[
\Gamma_{M}(w_{0})=\{v\in W:w_{0}R_{M}v\}.
\]
There are four possible evidence sets: 
\[
\varnothing,\qquad\{w_{0}\},\qquad\{w_{1}\},\qquad\{w_{0},w_{1}\}.
\]
For a crisp proposition $p$, the modal values at $w_{0}$ are: 
\[
\begin{array}{c|c|c|c|l}
\Gamma_{M}(w_{0}) & Mp(w_{0}) & \Diamond_{M}p(w_{0}) & \bar{M}p(w_{0}) & \text{Interpretation}\\
\hline \varnothing & 1 & 0 & 0 & \text{vacuous support; usually excluded by seriality}\\
\{w_{0}\} & p_{0} & p_{0} & p_{0} & \text{only the actual state matters}\\
\{w_{1}\} & p_{1} & p_{1} & p_{1} & \text{only the alternative state matters}\\
\{w_{0},w_{1}\} & p_{0}\wedge p_{1} & p_{0}\vee p_{1} & p_{0}\vee p_{1} & \text{both states matter}
\end{array}
\]
In the crisp Boolean case, 
\[
\Diamond_{M}p=\bar{M}p:=\neg M\neg p.
\]
Thus live possibility and non-exclusion coincide. The empty evidence
set is not appropriate for risk governance: it makes every claim vacuously
supported while making every claim impossible. This is why belief-like
systems usually impose seriality.

\subsubsection*{Knowledge-like assurance}

For knowledge-like assurance, the actual state should be included
in the evidence set. Hence the natural choices are 
\[
\Gamma_{K}(w_{0})=\{w_{0}\}\qquad\text{or}\qquad\Gamma_{K}(w_{0})=\{w_{0},w_{1}\}.
\]
If $\Gamma_{K}(w_{0})=\{w_{0}\}$, then 
\[
Kp(w_{0})=p_{0}.
\]
This makes every true claim at $w_{0}$ known at $w_{0}$, leaving
no room for true but unendorsed risk. The more informative case is
\[
\Gamma_{K}(w_{0})=\{w_{0},w_{1}\}.
\]
Then 
\[
Kp(w_{0})=p_{0}\wedge p_{1},\qquad\Diamond_{K}p(w_{0})=p_{0}\vee p_{1}.
\]
The four crisp possibilities are: 
\[
\begin{array}{c|c|c|c|l}
(p_{0},p_{1}) & Kp(w_{0}) & \Diamond_{K}p(w_{0}) & p_{0}\wedge\neg Kp(w_{0}) & \text{Governance reading}\\
\hline (0,0) & 0 & 0 & 0 & \text{\ensuremath{p} is excluded}\\
(1,0) & 0 & 1 & 1 & \text{actual risk, but not assurance-endorsed}\\
(0,1) & 0 & 1 & 0 & \text{not actual, but live nearby}\\
(1,1) & 1 & 1 & 0 & \text{robust assurance-grade endorsement}
\end{array}
\]
The second row is the canonical Moorean case: 
\[
p\wedge\neg Kp.
\]
The risk is real at $w_{0}$, but it is not assurance-endorsed because
the relevant alternative $w_{1}$ does not support it. Institutionally,
this is a fragile risk claim: it is true in the actual case, but it
fails across the assurance evidence set.

The third row is the possibility-only case: 
\[
\Diamond_{K}p\wedge\neg p.
\]
The risk is not actual at $w_{0}$, but it remains live at $w_{1}$.
This may justify monitoring, challenger modelling, or further investigation.

\subsubsection*{Belief-like working commitment}

For a belief-like working commitment, factivity is not required. The
evidence set may exclude the actual world. For example, take 
\[
\Gamma_{B}(w_{0})=\{w_{1}\}.
\]
Then 
\[
Bp(w_{0})=p_{1},\qquad\Diamond_{B}p(w_{0})=p_{1}.
\]
If 
\[
(p_{0},p_{1})=(0,1),
\]
then 
\[
Bp(w_{0})=1\qquad\text{but}\qquad p(w_{0})=0.
\]
This is non-factive working commitment. The institution acts as if
$p$ is decision-relevant because $p$ holds in the stress-relevant
alternative, even though $p$ is not actual.

If instead 
\[
(p_{0},p_{1})=(1,0),
\]
then 
\[
p(w_{0})=1\qquad\text{but}\qquad Bp(w_{0})=0.
\]
This is the belief-like Moorean diagnostic: 
\[
p\wedge\neg Bp.
\]
The risk is real, but the working-commitment register does not contain
it. Thus the same two-world structure represents both precautionary
commitment and failure of action-guiding recognition.

\subsubsection*{Two-world S5 and KD45 frames}

In the two-world setting, the S5 knowledge frames are the equivalence
relations on $W$. There are two main possibilities: 
\[
\begin{array}{c|c|l}
\text{Frame} & \Gamma_{K}(w) & \text{Interpretation}\\
\hline \text{identity} & \Gamma_{K}(w_{0})=\{w_{0}\},\ \Gamma_{K}(w_{1})=\{w_{1}\} & \text{each world is perfectly distinguished}\\
\text{universal} & \Gamma_{K}(w_{0})=\Gamma_{K}(w_{1})=\{w_{0},w_{1}\} & \text{the two worlds are indistinguishable}
\end{array}
\]
The identity frame makes knowledge too strong for most governance
applications: if $p$ is true, then $Kp$ is true. The universal frame
is more useful for modelling assurance under uncertainty: $Kp$ holds
only when $p$ is true in both the actual and alternative states.

For KD45 belief, the relation is serial, transitive, and Euclidean.
On two worlds, the main possibilities are: 
\[
\begin{array}{c|c|l}
\text{Evidence sets} & \text{Factive?} & \text{Interpretation}\\
\hline \Gamma_{B}(w_{0})=\Gamma_{B}(w_{1})=\{w_{0}\} & \text{not everywhere} & \text{all belief is anchored on }w_{0}\\
\Gamma_{B}(w_{0})=\Gamma_{B}(w_{1})=\{w_{1}\} & \text{not everywhere} & \text{all belief is anchored on }w_{1}\\
\Gamma_{B}(w_{0})=\{w_{0}\},\ \Gamma_{B}(w_{1})=\{w_{1}\} & \text{yes} & \text{identity case}\\
\Gamma_{B}(w_{0})=\Gamma_{B}(w_{1})=\{w_{0},w_{1}\} & \text{yes} & \text{universal case}
\end{array}
\]
The non-factive KD45 cases are useful for working commitment. They
allow the institution to act on a stress state or decision-relevant
scenario that is not the actual state.

\subsubsection*{Fuzzy two-world case}

Now let the evidence relation be fuzzy. Write 
\[
\gamma_{M}=\begin{pmatrix}a & b\\
c & d
\end{pmatrix},
\]
where 
\[
a=\gamma_{M}(w_{0},w_{0}),\quad b=\gamma_{M}(w_{0},w_{1}),\quad c=\gamma_{M}(w_{1},w_{0}),\quad d=\gamma_{M}(w_{1},w_{1}).
\]
Let 
\[
p(w_{0})=x,\qquad p(w_{1})=y.
\]
Using the Gödel/min package, 
\[
(Mp)(w_{0})=\min\{a\Rightarrow x,\ b\Rightarrow y\},
\]
and 
\[
(\Diamond_{M}p)(w_{0})=\max\{\min(a,x),\min(b,y)\}.
\]
Similarly, 
\[
(Mp)(w_{1})=\min\{c\Rightarrow x,\ d\Rightarrow y\},
\]
and 
\[
(\Diamond_{M}p)(w_{1})=\max\{\min(c,x),\min(d,y)\}.
\]
With standard negation, the dual non-exclusion operator at $w_{0}$
is 
\[
(\bar{M}p)(w_{0})=1-\min\{a\Rightarrow(1-x),\ b\Rightarrow(1-y)\}.
\]
The hesitation margin at $w_{0}$ is 
\[
H_{M}(p)(w_{0})=\max\{0,\bar{M}p(w_{0})-Mp(w_{0})\}.
\]
This fuzzy case shows why $\Diamond_{M}p$ and $\bar{M}p$ should
be kept conceptually distinct. In crisp Boolean settings they coincide.
In fuzzy settings, they may differ. The operator $\Diamond_{M}p$
measures the degree to which $p$ is live somewhere in the evidence
field. The operator $\bar{M}p$ measures the degree to which the support
operator fails to rule out $p$.

\subsubsection*{Varying $p_{0}$ and $p_{1}$}

A useful way to visualise the fuzzy two-world case is to let 
\[
p_{0}=p(w_{0}),\qquad p_{1}=p(w_{1})
\]
range over the whole unit square $[0,1]^{2}$. Fix, for illustration,
\[
\gamma_{K}(w_{0},w_{0})=1,\qquad\gamma_{K}(w_{0},w_{1})=0.6,
\]
so that the assurance standard treats the actual state as fully relevant
and the alternative state as partially relevant. Also fix 
\[
\gamma_{B}(w_{0},w_{0})=0,\qquad\gamma_{B}(w_{0},w_{1})=1,
\]
so that the working-commitment standard focuses on the stress state.

Under these choices, the following quantities become functions on
the $(p_{0},p_{1})$-plane: 
\[
Kp(w_{0}),\qquad\Diamond_{K}p(w_{0}),\qquad\bar{K}p(w_{0}),\qquad H_{K}(p)(w_{0}),
\]
\[
p(w_{0})\wedge\neg Kp(w_{0}),\qquad p(w_{0})\wedge K\neg p(w_{0}),\qquad I_{K}(p)(w_{0}),
\]
and 
\[
Bp(w_{0}),\qquad p(w_{0})\wedge\neg Bp(w_{0}),\qquad\max\{0,Bp(w_{0})-p(w_{0})\}.
\]

These plots show how each diagnostic changes as the actual-state degree
$p_{0}$ and the alternative-state degree $p_{1}$ vary. For instance,
the Moorean knowledge diagnostic 
\[
p(w_{0})\wedge\neg Kp(w_{0})
\]
is strongest when $p_{0}$ is high but the $K$-standard does not
endorse $p$. This occurs when the risk is strong at the actual state
but weak, unstable, or insufficiently supported at the relevant alternative
state. By contrast, the non-factive belief gap 
\[
\max\{0,Bp(w_{0})-p(w_{0})\}
\]
is strongest when the stress state supports $p$ more strongly than
the actual state. This is the formal pattern of precautionary working
commitment.

The following figures illustrate the main $K$-side quantities. The
contour lines mark the levels 0.25, 0.5, and 0.75 of the corresponding
modal diagnostic.

\begin{figure}[h]
\centering \includegraphics[width=0.45\textwidth]{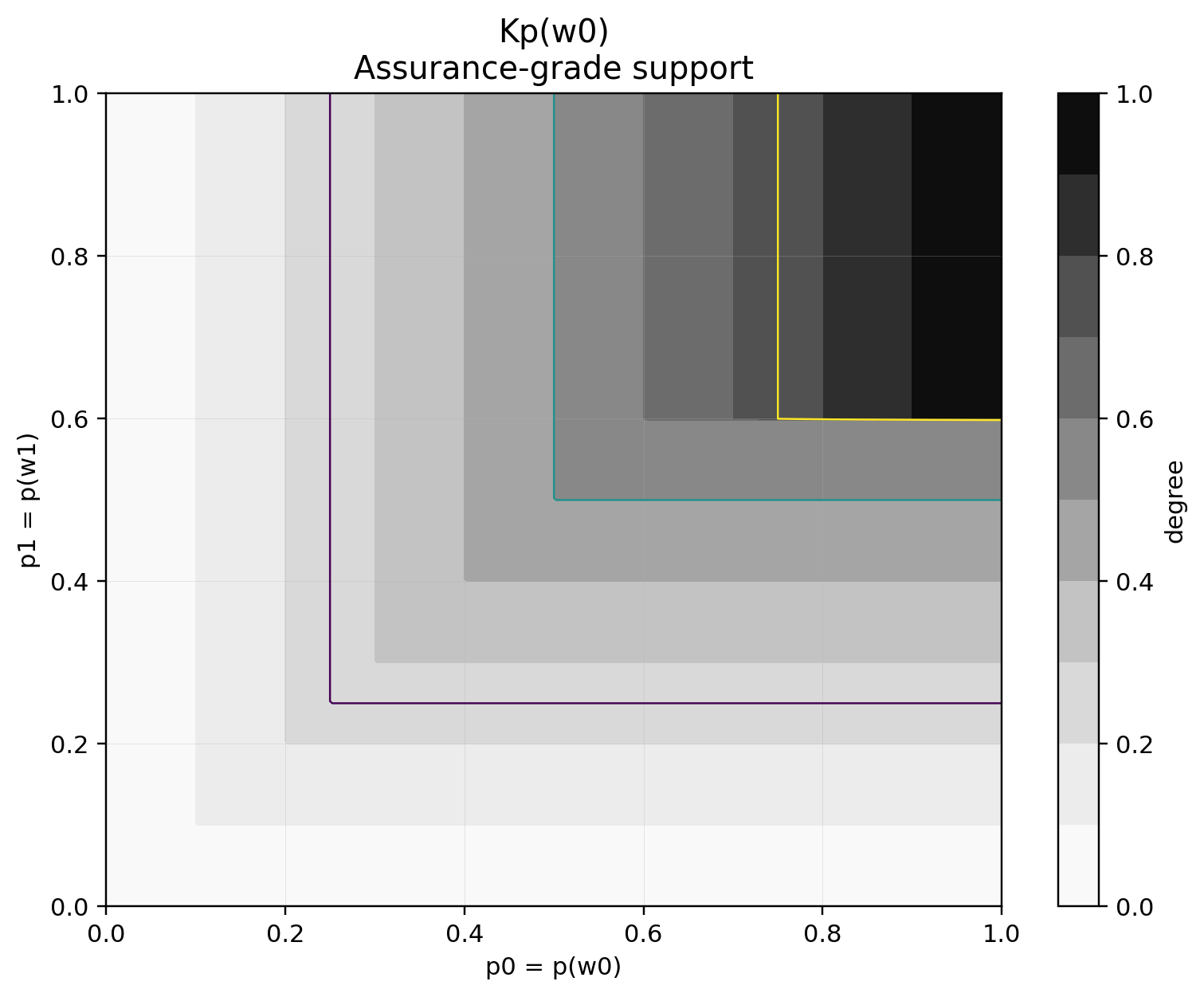}
\includegraphics[width=0.45\textwidth]{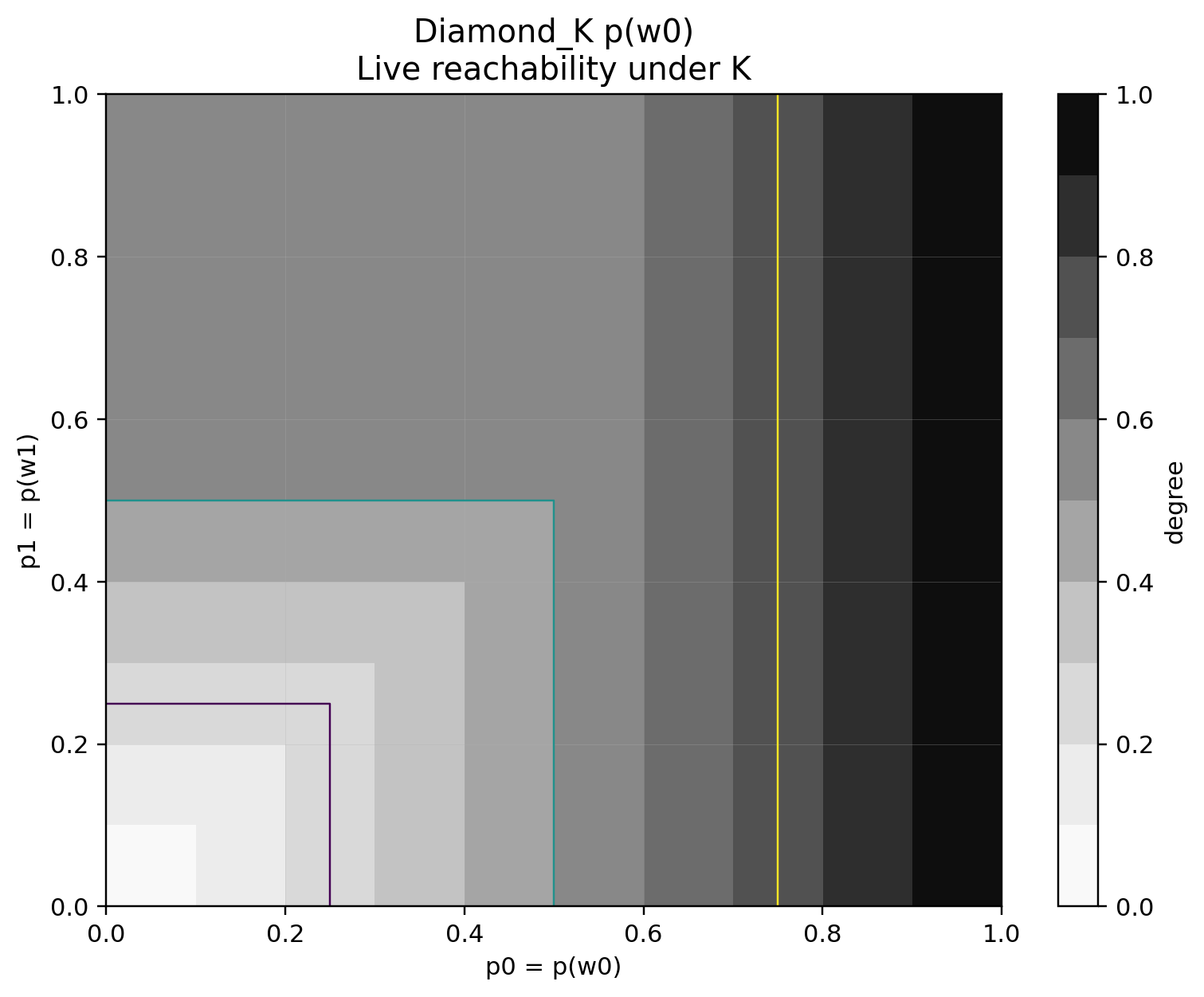}
\caption{Two-world fuzzy case. Left: assurance-grade support $Kp(w_{0})$.
Right: live reachability $\Diamond_{K}p(w_{0})$.}
\label{fig:two-world-K-diamond}
\end{figure}

\begin{figure}[h]
\centering \includegraphics[width=0.45\textwidth]{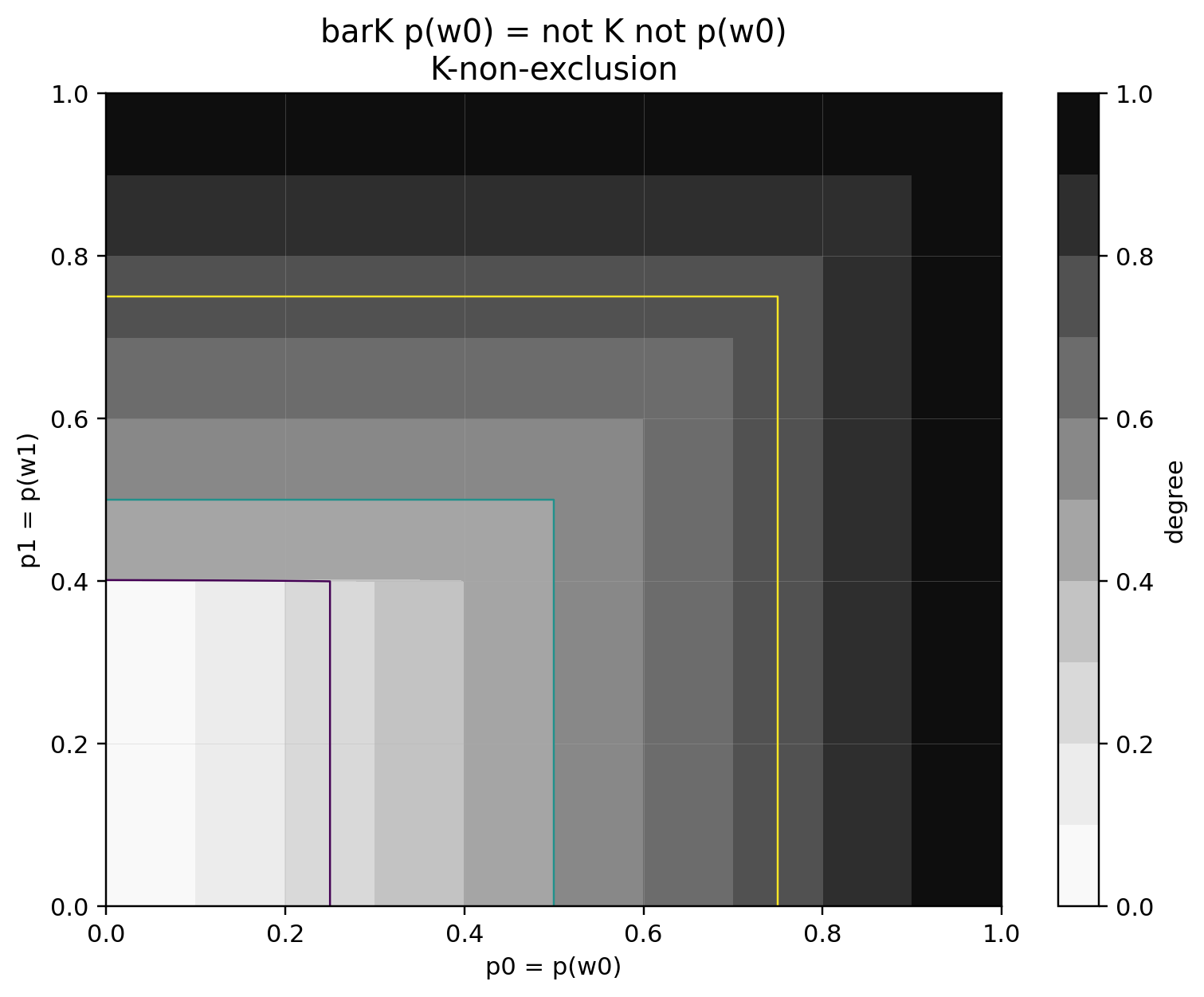}
\includegraphics[width=0.45\textwidth]{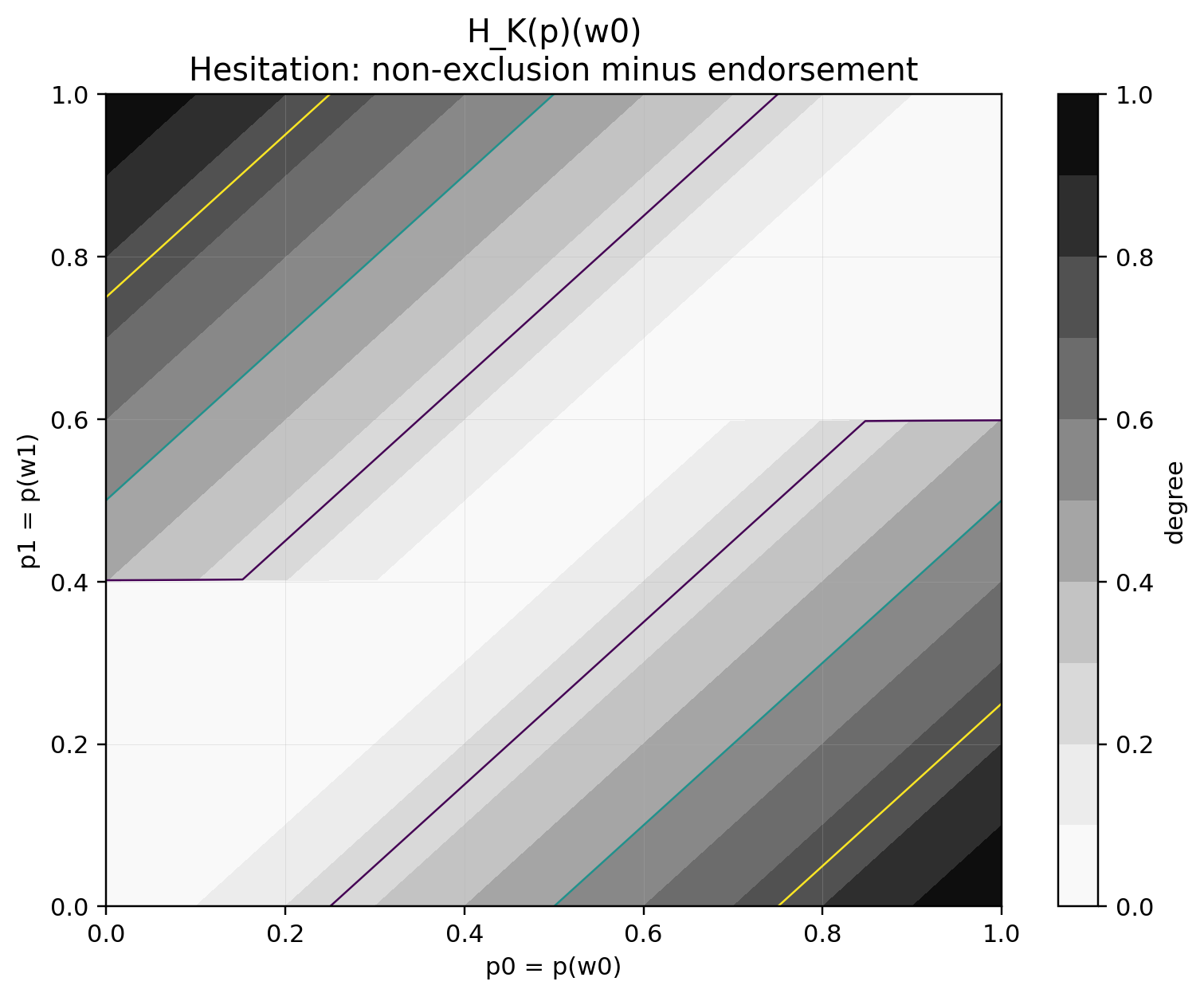}
\caption{Two-world fuzzy case. Left: non-exclusion $\bar{K}p(w_{0})$. Right:
hesitation $H_{K}(p)(w_{0})$.}
\label{fig:two-world-bark-hk}
\end{figure}

The next figures show the main diagnostics.

\begin{figure}[h]
\centering \includegraphics[width=0.45\textwidth]{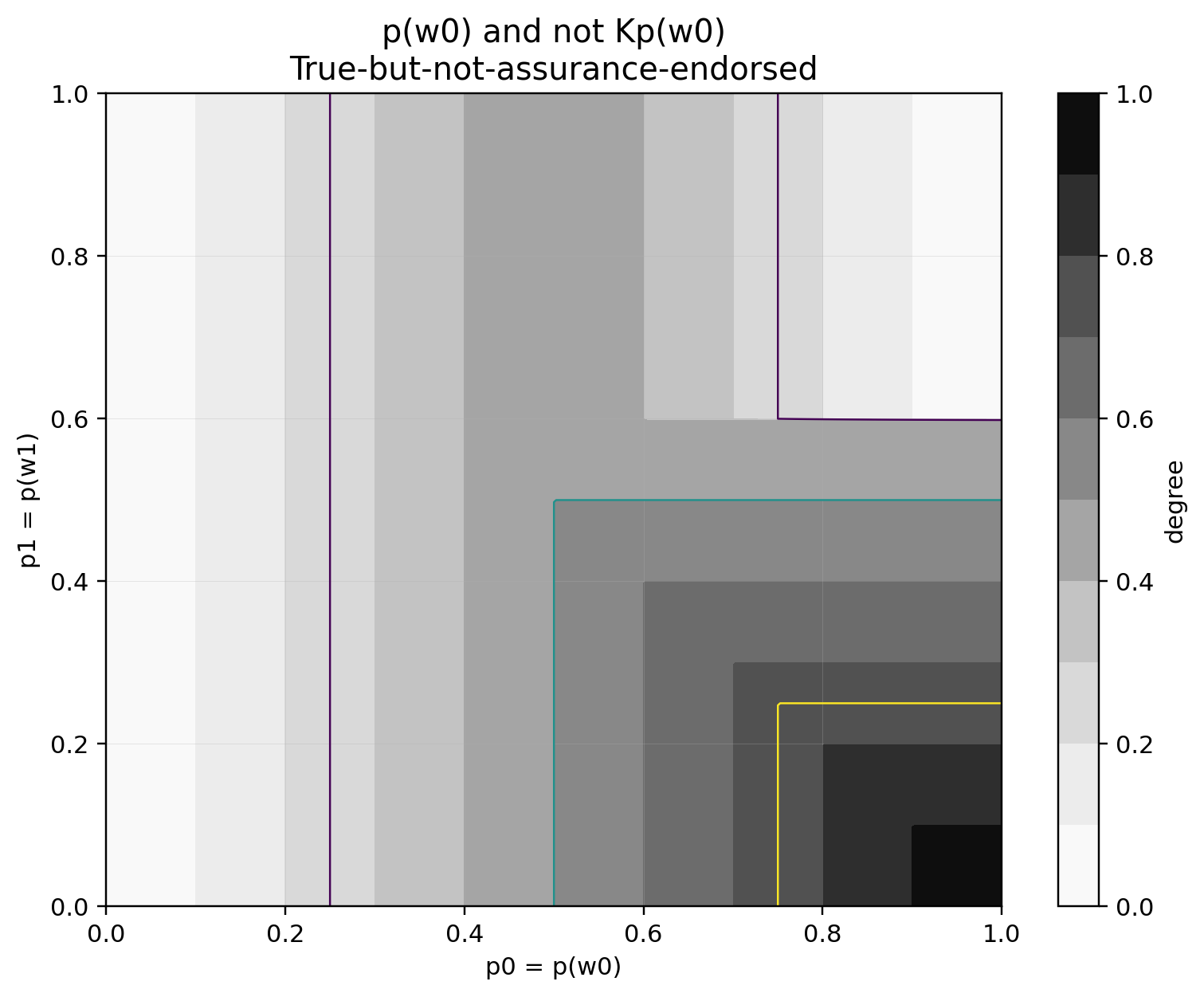}
\includegraphics[width=0.45\textwidth]{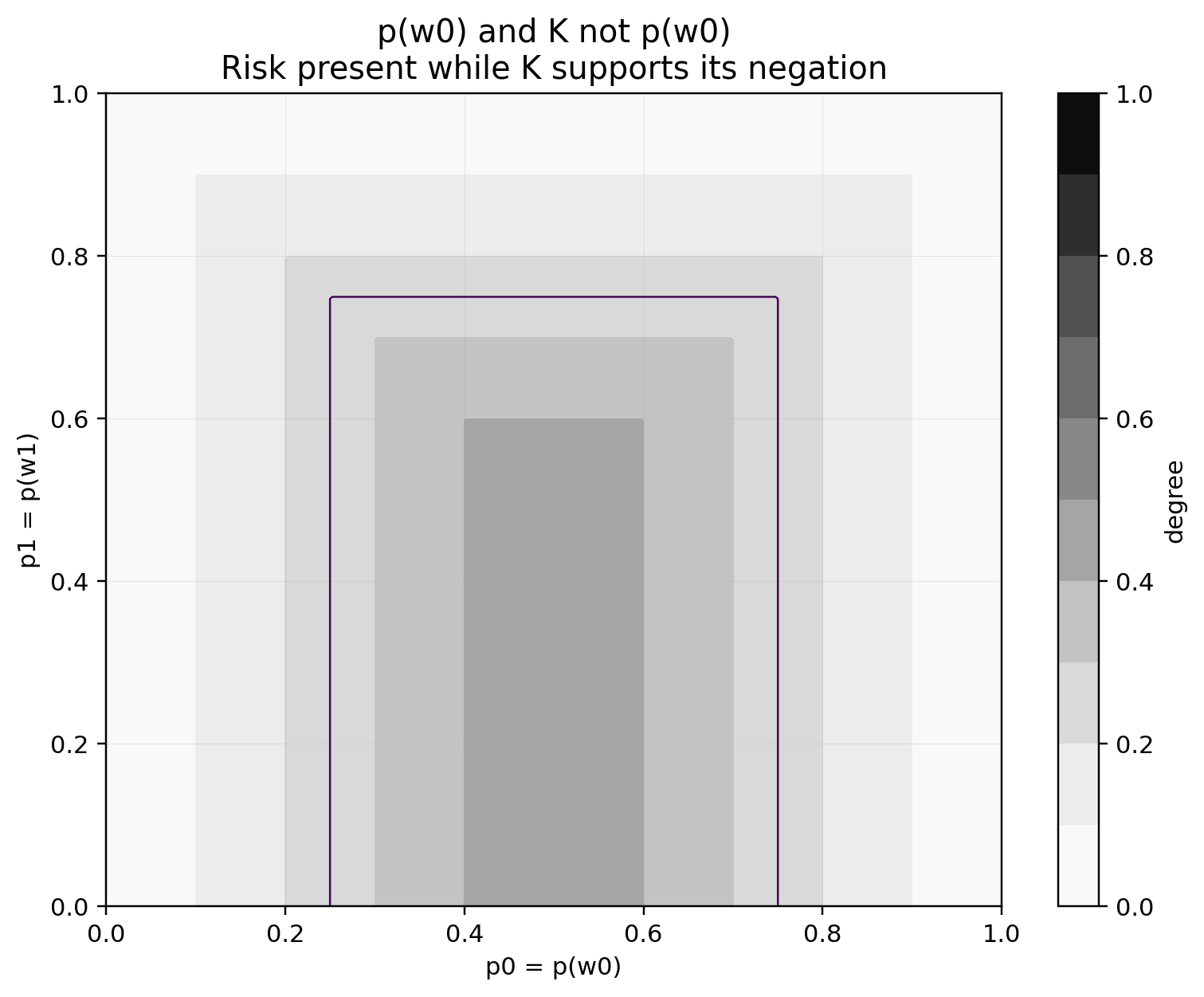}
\caption{Two-world fuzzy case. Left: Moorean diagnostic $p(w_{0})\wedge\neg Kp(w_{0})$.
Right: anti-Moorean diagnostic $p(w_{0})\wedge K\neg p(w_{0})$.}
\label{fig:two-world-moorean-anti}
\end{figure}

\begin{figure}[h]
\centering \includegraphics[width=0.45\textwidth]{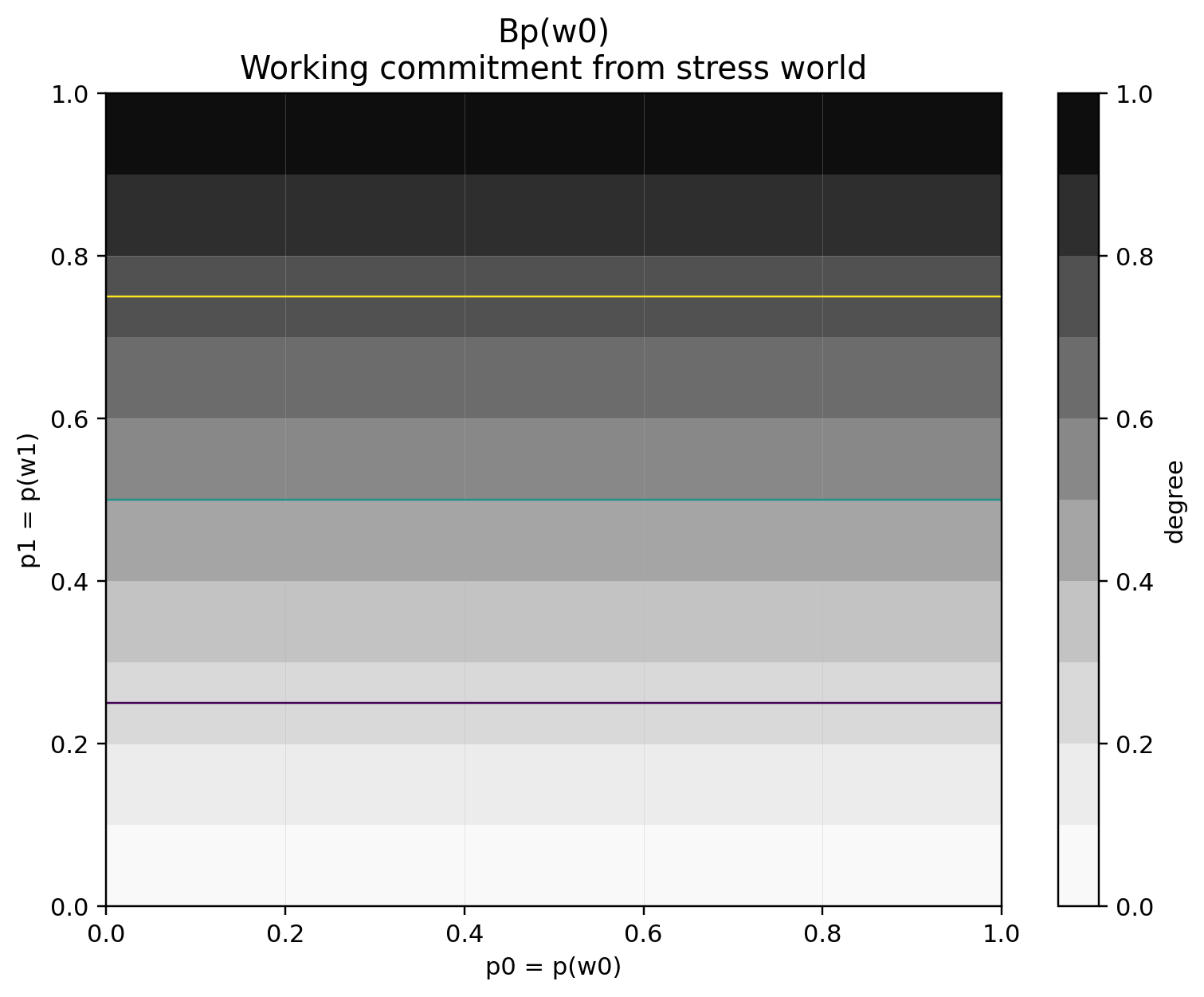}
\includegraphics[width=0.45\textwidth]{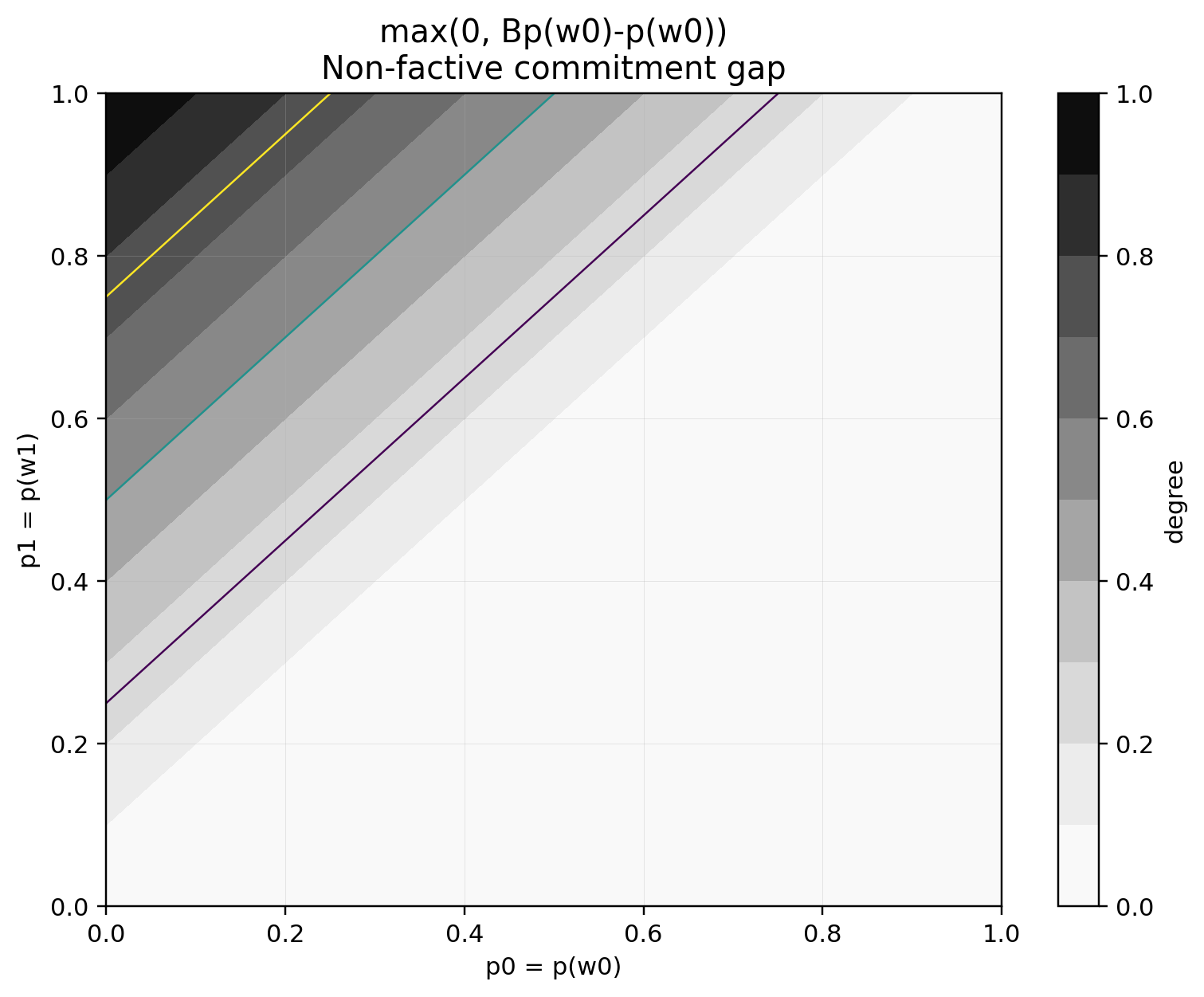}
\caption{Two-world fuzzy case. Left: working commitment $Bp(w_{0})$. Right:
non-factive commitment gap $\max\{0,Bp(w_{0})-p(w_{0})\}$.}
\label{fig:two-world-B-gap}
\end{figure}

\subsubsection*{Numerical application: model risk and stress commitment}

Consider a two-world risk-governance setting in which $w_{0}$ is
the current institutional model state and $w_{1}$ is a challenger
or stress state.

First, let 
\[
p=\text{``the \ensuremath{99\%} expected shortfall breaches the capital threshold.''}
\]
Suppose the capital threshold is 
\[
c=100,
\]
and the two model states give 
\[
\mathrm{ES}_{0.99}(w_{0})=112,\qquad\mathrm{ES}_{0.99}(w_{1})=96.
\]
Thus, in crisp form, 
\[
p(w_{0})=1,\qquad p(w_{1})=0.
\]
The current model state breaches the threshold, but the challenger
state does not. Let 
\[
\Gamma_{K}(w_{0})=\{w_{0},w_{1}\}.
\]
Then 
\[
Kp(w_{0})=0,\qquad\Diamond_{K}p(w_{0})=1,\qquad\bar{K}p(w_{0})=1.
\]
Hence 
\[
p(w_{0})\wedge\neg Kp(w_{0})=1.
\]
The interpretation is direct: the breach occurs in the current model
state, but the breach conclusion is not robust across the assurance
evidence set. This is a Moorean model-risk diagnostic.

The corresponding table is: {\scriptsize
\[
\begin{array}{c|c|c|c|c|l}
p(w_{0}) & p(w_{1}) & Kp(w_{0}) & \Diamond_{K}p(w_{0}) & p\wedge\neg Kp(w_{0}) & \text{Reading}\\
\hline 1 & 0 & 0 & 1 & 1 & \text{breach occurs but lacks assurance robustness}
\end{array}
\]
}{\scriptsize\par}

Now consider a second proposition: 
\[
q=\text{``a systemic liquidity cascade occurs.''}
\]
Suppose 
\[
q(w_{0})=0,\qquad q(w_{1})=1.
\]
Thus no cascade occurs in the current state, but a cascade occurs
in the stress state. Let 
\[
\Gamma_{B}(w_{0})=\{w_{1}\}.
\]
Then 
\[
Bq(w_{0})=1,\qquad\Diamond_{B}q(w_{0})=1,\qquad\bar{B}q(w_{0})=1,
\]
while 
\[
q(w_{0})=0.
\]
This is non-factive working commitment. The institution is not claiming
that a cascade is actual. It is adopting the cascade scenario as action-guiding
because it holds in the stress-relevant world.

The table is: {\scriptsize
\[
\begin{array}{c|c|c|c|c|l}
q(w_{0}) & q(w_{1}) & Bq(w_{0}) & \Diamond_{B}q(w_{0}) & \bar{B}q(w_{0}) & \text{Reading}\\
\hline 0 & 1 & 1 & 1 & 1 & \text{not actual, but adopted as a precautionary commitment}
\end{array}
\]
}{\scriptsize\par}

\subsubsection*{Fuzzy numerical application}

Finally, consider a fuzzy liquidity case. Let 
\[
r=\text{``the bank breaches the liquidity threshold.''}
\]
Suppose 
\[
r(w_{0})=0,\qquad r(w_{1})=0.9.
\]
Thus there is no breach in the current state, but a strong breach
in the stress state. Let 
\[
\gamma_{K}(w_{0},w_{0})=1,\qquad\gamma_{K}(w_{0},w_{1})=0.6.
\]
Using the Gödel/min package, 
\[
(Kr)(w_{0})=\min\{1\Rightarrow0,\;0.6\Rightarrow0.9\}=0.
\]
The live possibility is 
\[
(\Diamond_{K}r)(w_{0})=\max\{\min(1,0),\min(0.6,0.9)\}=0.6.
\]
With standard negation, 
\[
(K\neg r)(w_{0})=\min\{1\Rightarrow1,\;0.6\Rightarrow0.1\}=0.1.
\]
Therefore 
\[
\bar{K}r(w_{0})=1-K\neg r(w_{0})=0.9.
\]
The hesitation margin is 
\[
H_{K}(r)(w_{0})=\max\{0,\bar{K}r(w_{0})-Kr(w_{0})\}=0.9.
\]
The fuzzy table is: {\scriptsize
\[
\begin{array}{c|c|c|c|c|c|l}
r(w_{0}) & r(w_{1}) & Kr(w_{0}) & \Diamond_{K}r(w_{0}) & \bar{K}r(w_{0}) & H_{K}(r)(w_{0}) & \text{Reading}\\
\hline 0 & 0.9 & 0 & 0.6 & 0.9 & 0.9 & \text{not endorsed, live, and not ruled out}
\end{array}
\]
}This example shows why $\Diamond_{M}p$ and $\bar{M}p$ should be
distinguished in fuzzy semantics: 
\[
\Diamond_{K}r(w_{0})=0.6,\qquad\bar{K}r(w_{0})=0.9.
\]
The first measures live reachability through the evidence relation.
The second measures failure to rule out $r$ through the dual of the
support operator.

If the working-commitment standard focuses only on the stress world,
\[
\gamma_{B}(w_{0},w_{0})=0,\qquad\gamma_{B}(w_{0},w_{1})=1,
\]
then 
\[
(Br)(w_{0})=\min\{0\Rightarrow0,\;1\Rightarrow0.9\}=0.9.
\]
Thus the institution has a strong working commitment to the liquidity-breach
scenario, even though 
\[
r(w_{0})=0.
\]
Again, $B$ is non-factive: 
\[
Br(w_{0})>r(w_{0}).
\]

The governance interpretation is clear. Under the assurance standard
$K$, the breach is not certified. Under the working-commitment standard
$B$, the stress breach is strong enough to guide action. The same
two-world structure therefore represents assurance failure, live possibility,
non-exclusion, hesitation, and precautionary commitment.

\subsection{Model risk: assurance under model variation}

\label{sec:model-risk-example}

Consider a financial institution using quantitative models to estimate
tail loss. The first-order question is whether the estimated loss
breaches a capital, solvency, or risk-appetite threshold. The second-order
question is whether that conclusion is stable enough for institutional
reliance.

Let a model-evaluation state be 
\[
w=(D,\mathcal{M},\theta),
\]
where $D$ is the dataset, $\mathcal{M}$ is the model class, and
$\theta$ is the parameter vector. Let $L^{(w)}$ be the loss variable
implied by $w$. For confidence level $\alpha$, define 
\[
\mathrm{VaR}_{\alpha}(w):=\inf\{x:\mathbb{P}_{\mathcal{M},\theta}(L^{(w)}\le x)\ge\alpha\},
\]
and 
\[
\mathrm{ES}_{\alpha}(w):=\mathbb{E}_{\mathcal{M},\theta}\bigl[L^{(w)}\mid L^{(w)}\ge\mathrm{VaR}_{\alpha}(w)\bigr].
\]

Let $c$ be a capital or risk-appetite threshold. A natural object-level
risk proposition is 
\[
p(w)=\mathbf{1}_{\{\mathrm{ES}_{\alpha}(w)\ge c\}},
\]
or, in a graded version, 
\[
p(w)=\psi\bigl(\mathrm{ES}_{\alpha}(w)-c\bigr),
\]
where $\psi:\mathbb{R}\to[0,1]$ is increasing. Thus $p(w)$ measures
the degree to which model state $w$ indicates a tail-loss breach.

To represent assurance under model variation, define an evidence relation
between model states by 
\[
\gamma_{K}(w,v)=\exp\Bigl(-\lambda_{D}d_{D}(D,D')-\lambda_{\mathcal{M}}\mathbf{1}_{\mathcal{M}\neq\mathcal{M}'}-\lambda_{\theta}\|\theta-\theta'\|^{2}\Bigr),
\]
where 
\[
v=(D',\mathcal{M}',\theta').
\]
The value $\gamma_{K}(w,v)$ is high when $v$ is a nearby admissible
model state relative to $w$. Nearby states may differ in data cleaning,
sampling window, model class, parameter estimate, or implementation.
The weights 
\[
\lambda_{D},\qquad\lambda_{\mathcal{M}},\qquad\lambda_{\theta}
\]
encode how sensitive the assurance standard is to each source of variation.

The current model output is $p(w)$. The assurance-grade status is
$Kp(w)$. Thus $Kp(w)$ asks whether the breach conclusion survives
across the relevant model neighbourhood. The live-reachability status
is $\Diamond_{K}p(w)$: it asks whether a breach remains possible
under admissible model variation. The dual status $\overline{K}p(w)$
asks whether a breach has not been ruled out by the assurance standard.
The hesitation margin $H_{K}(p)(w)$ measures the gap between non-exclusion
and endorsement.

The main regions are therefore: 
\[
Kp:\quad\text{the breach is robustly assurance-endorsed;}
\]
\[
p\wedge\neg Kp:\quad\text{the breach occurs, but lacks assurance-grade robustness;}
\]
\[
\Diamond_{K}p\setminus p:\quad\text{there is no breach at the current state, but breach is nearby;}
\]
\[
\Diamond_{K}p\wedge\neg Kp:\quad\text{the breach lies in the hesitation region;}
\]
\[
\neg\Diamond_{K}p:\quad\text{breach is excluded under the current evidence standard.}
\]

We now specialise this general setup. Suppose the dataset $D$ is
fixed and the model class is fixed to the lognormal family: 
\[
L\sim\mathrm{Lognormal}(\mu,\sigma^{2}),
\]
with parameter 
\[
\theta=(\mu,\sigma).
\]
Other challenger classes could include gamma, Weibull, Student-$t$,
or a lognormal body with a generalized Pareto tail. The lognormal
case is used here because it gives a simple two-dimensional parameter
space.

For a lognormal loss variable, 
\[
\mathrm{VaR}_{\alpha}(\mu,\sigma)=\exp(\mu+\sigma z_{\alpha}),
\]
where $z_{\alpha}$ is the $\alpha$-quantile of the standard normal
distribution. The expected shortfall is 
\[
\mathrm{ES}_{\alpha}(\mu,\sigma)=\exp\left(\mu+\frac{\sigma^{2}}{2}\right)\frac{\Phi(\sigma-z_{\alpha})}{1-\alpha}.
\]
Define the crisp breach proposition 
\[
p(\mu,\sigma)=\mathbf{1}_{\{\mathrm{ES}_{0.99}(\mu,\sigma)\ge c\}}.
\]
Thus $p$ says that the fitted lognormal model implies a $99\%$ expected-shortfall
breach.

Because $D$ and $\mathcal{M}$ are fixed in this special case, only
$\theta=(\mu,\sigma)$ varies. Define the parameter-neighbourhood
distance by 
\[
d\bigl((\mu,\sigma),(\mu',\sigma')\bigr)^{2}=\left(\frac{\mu-\mu'}{\beta_{\mu}}\right)^{2}+\left(\frac{\sigma-\sigma'}{\beta_{\sigma}}\right)^{2}.
\]
The assurance relation is 
\[
\gamma_{K}(w,v)=1\quad\Longleftrightarrow\quad d(w,v)\le1,
\]
and $\gamma_{K}(w,v)=0$ otherwise, where 
\[
w=(\mu,\sigma),\qquad v=(\mu',\sigma').
\]
In this crisp case, 
\[
Kp(w)=1\quad\Longleftrightarrow\quad p(v)=1\text{ for every admissibly nearby }v,
\]
and 
\[
\Diamond_{K}p(w)=1\quad\Longleftrightarrow\quad p(v)=1\text{ for some admissibly nearby }v.
\]
Since the example is crisp Boolean, 
\[
\Diamond_{K}p=\overline{K}p.
\]

Figure~\ref{fig:model-risk-modal-regions} plots these regions for
\[
\alpha=0.99,\qquad c=100,\qquad\beta_{\mu}=0.10,\qquad\beta_{\sigma}=0.045.
\]

\begin{figure}[h]
\centering{}\centering \includegraphics[width=0.75\textwidth]{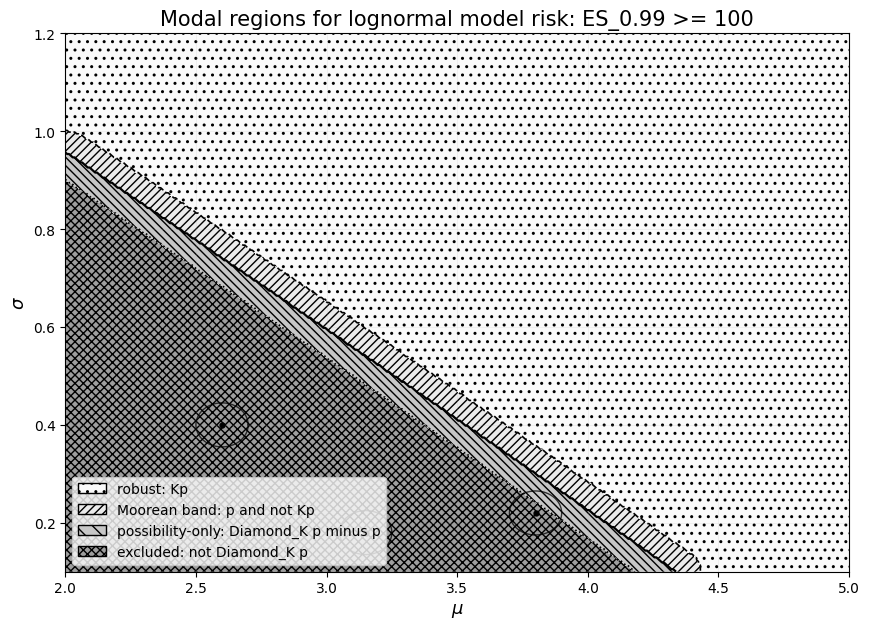}
\caption{Modal regions for the lognormal model-risk example. The robust region
is $Kp$, the Moorean band is $p\wedge\neg Kp$, the possibility-only
region is $\Diamond_{K}p\setminus p$, and the excluded region is
$\neg\Diamond_{K}p$.}
\label{fig:model-risk-modal-regions}
\end{figure}

The Moorean band is the main governance warning. It consists of parameter
states where the tail-loss breach occurs, but the conclusion is not
robust under admissible parameter perturbations. The possibility-only
band is also important: the current parameter does not breach the
threshold, but a nearby admissible parameter does. This may justify
challenger modelling, sensitivity analysis, model overlays, or temporary
model-use restrictions.

This example separates the model output from its institutional standing.
The proposition $p$ records whether the fitted model breaches the
threshold. The formula $Kp$ records whether the breach is robust
enough for assurance-grade use. The formula $\Diamond_{K}p$ records
whether breach remains live under admissible variation. Probability
estimates alone do not encode these distinctions. 

\subsection{Banking-network contagion: non-factive working commitment}

\label{subsec:contagion}

The banking-network example illustrates why the belief-like operator
$B$ should not be factive. A supervisor or risk committee may adopt
a working commitment to a severe-risk scenario even when the adverse
state is not occurring in the actual world.

Let 
\[
W=\{w_{0},w_{1},w_{2}\},
\]
where $w_{0}$ is the actual banking state and $w_{1},w_{2}$ are
supervisory stress states. Let 
\[
p=\text{``a systemic cascade occurs.''}
\]
Suppose 
\[
p(w_{0})=0,\qquad p(w_{1})=1,\qquad p(w_{2})=1.
\]
Thus no cascade occurs at the actual state, but a cascade occurs in
both stress-relevant states.

Define the belief evidence relation at $w_{0}$ by 
\[
\gamma_{B}(w_{0},w_{0})=0,
\]
and 
\[
\gamma_{B}(w_{0},w_{1})=1,\qquad\gamma_{B}(w_{0},w_{2})=1.
\]
Equivalently, 
\[
\Gamma_{B}(w_{0})=\{w_{1},w_{2}\}.
\]
The actual state $w_{0}$ is not included in the $B$-evidence set.
This is the formal reason $B$ is not factive in this example.

Using the G\"odel implication, 
\[
(Bp)(w_{0})=\inf_{v\in W}\bigl(\gamma_{B}(w_{0},v)\Rightarrow p(v)\bigr).
\]
Since the only fully relevant $B$-worlds are $w_{1}$ and $w_{2}$,
and $p$ holds at both of them, 
\[
(Bp)(w_{0})=1.
\]
But 
\[
p(w_{0})=0.
\]
Therefore 
\[
(Bp)(w_{0})>p(w_{0}),
\]
so $B$ is non-factive.

This is not a defect of the model. It is the intended behaviour of
a precautionary working-commitment operator. The institution is not
claiming that a systemic cascade is actually occurring. Rather, it
is treating the cascade scenario as action-guiding because it occurs
throughout the relevant supervisory stress set.

This captures a common feature of systemic-risk governance. A regulator,
central bank, or risk committee may act on a contagion concern before
the contagion materialises. It may require liquidity buffers, restrict
distributions, demand stress testing, or intensify supervision. These
actions do not require the assertion that the adverse state is actual.
They require a disciplined commitment to the scenario as relevant
for decision-making.

The example also clarifies the difference between $K$ and $B$. If
the operator were $K$, factivity would require the actual world to
be answerable to the knowledge claim. Since $p(w_{0})=0$, one could
not have $Kp(w_{0})=1$. But $B$ is designed for a different governance
role. It permits action under incomplete assurance, especially where
severity is high and delay is costly.

The important distinction is therefore: 
\[
Kp:\quad\text{the institution certifies that \ensuremath{p} is correct;}
\]
\[
Bp:\quad\text{the institution is licensed to act as if \ensuremath{p} is decision-relevant.}
\]
In banking-network contagion, this distinction is essential. Systemic
crises often develop through feedback loops, liquidity hoarding, interbank
exposures and correlated funding stress. Waiting for factive knowledge
may mean waiting until the cascade has already begun. The belief-like
operator $B$ represents the governance space between ignorance and
certified knowledge: the space of disciplined precaution.

\begin{figure}
\begin{centering}
\includegraphics[scale=0.7]{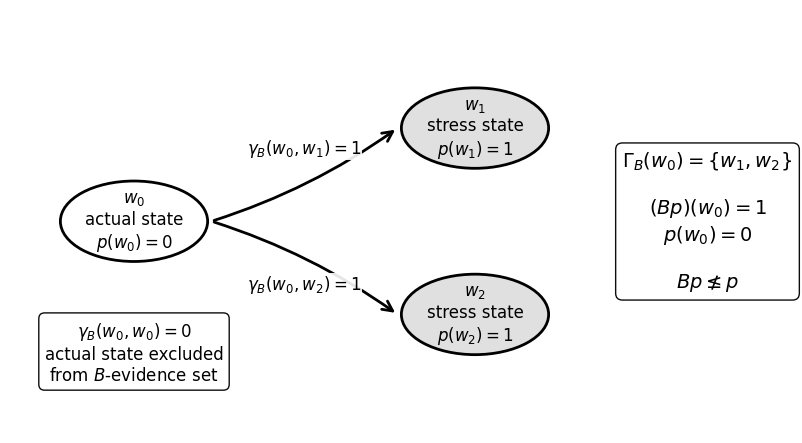}
\par\end{centering}
\caption{A concrete non-factive working-commitment example. The actual state
$w_{0}$ does not satisfy the systemic-cascade proposition $p$, so
$p(w_{0})=0$. However, the $B$-evidence set at $w_{0}$ excludes
the actual state and contains the two supervisory stress states $w_{1}$
and $w_{2}$, both of which satisfy $p$. Hence $(Bp)(w_{0})=1$ while
$p(w_{0})=0$, showing that $B$ is non-factive: $Bp\nleq p$. The
operator $B$ therefore represents a precautionary working commitment
rather than truth-entailing assurance.}
\end{figure}

The banking-network example draws on work on financial contagion,
clearing networks, and systemic risk \cite{AllenGale2000Financial,EisenbergNoe2001Systemic,ElliottGolubJackson2014Financial,AcemogluOzdaglarTahbazSalehi2015Systemic,GlassermanYoung2016Contagion}.
In that literature, shocks propagate through exposure networks. The
present framework abstracts from the detailed clearing mechanism and
represents supervisory relevance by the evidence profile $\gamma_{B}(w,v)$.

\subsection{A simplified two-dimensional flood model}

\label{subsec:flood}

The flood example shows how the same modal framework applies to a
physical risk system. Here the object-level claim concerns whether
a location is flood-critical. The epistemic question concerns whether
that classification is robust under nearby hydrological or infrastructural
states.

Let a world be a flood-risk state $w=(X,Y),$ where $X$ is rainfall
intensity and $Y$ is drainage blockage. Let $F:W\to\mathbb{R}$ be
a hydraulic stress index. It may represent overtopping pressure, expected
inundation, drainage overload, or a simplified combined stress score.
Let $c$ be a critical threshold. Define the crisp flood proposition
\[
p(w)=\begin{cases}
1, & F(w)\ge c,\\
0, & F(w)<c.
\end{cases}
\]
Thus $p$ says that the state is flood-critical.

Now define an evidence relation using a weighted tolerance ball. For
\[
w=(X,Y),\qquad v=(X',Y'),
\]
let 
\[
d(w,v)^{2}=a_{X}|X-X'|^{2}+a_{Y}|Y-Y'|^{2},
\]
where $a_{X},a_{Y}>0$. The weights encode the relative importance
of rainfall and drainage blockage. For a tolerance level $\beta>0$,
define 
\[
\gamma_{K}(w,v)=\begin{cases}
1, & d(w,v)\le\beta,\\
0, & d(w,v)>\beta.
\end{cases}
\]
The corresponding evidence neighbourhood is 
\[
\Gamma_{K}(w)=\{v\in W:d(w,v)\le\beta\}.
\]

In the crisp case, 
\[
Kp(w)=1\Longleftrightarrow p(v)=1\text{ for all }v\in\Gamma_{K}(w),
\]
and 
\[
\Diamond_{K}p(w)=1\Longleftrightarrow p(v)=1\text{ for some }v\in\Gamma_{K}(w).
\]
Thus $Kp$ says that flood-criticality is robust throughout the evidence
neighbourhood. By contrast, $\Diamond_{K}p$ says that flood-criticality
occurs somewhere in that neighbourhood.

The Moorean diagnostic 
\[
p\wedge\neg Kp
\]
then says that the actual state is flood-critical, but the flood classification
is not robust. This occurs near the boundary of the flood region.
At such a state, flooding occurs, but a nearby admissible change in
rainfall intensity, drainage blockage, saturation, tide, or calibration
may move the state outside the flood-critical region.

To visualise the structure, take 
\[
W=[0,1]^{2},
\]
where $x$ denotes rainfall intensity and $y$ denotes drainage blockage.
Let 
\[
F(x,y)=0.95x^{2}+0.75y^{2}+0.85xy+0.10\sin(2.5\pi x)\sin(2\pi y).
\]
Define 
\[
A:=\{(x,y)\in W:F(x,y)\ge c\}.
\]
Then $A$ is the flood-critical region. The proposition $p$ is the
indicator of $A$.

For a radius $\beta$, the evidence neighbourhood of $w$ is 
\[
\Gamma_{K}(w)=\{v\in W:d(w,v)\le\beta\}.
\]
Then 
\[
Kp(w)=1\Longleftrightarrow\Gamma_{K}(w)\subseteq A.
\]
So $Kp$ is the robust interior of the flood region. A point deep
inside $A$ satisfies $Kp$, because all nearby admissible states
are also flood-critical.

Similarly, 
\[
\Diamond_{K}p(w)=1\Longleftrightarrow\Gamma_{K}(w)\cap A\neq\varnothing.
\]
So $\Diamond_{K}p$ is the outer expansion of the flood region. A
point outside $A$ may still satisfy $\Diamond_{K}p$ if a nearby
admissible state is flood-critical.

This gives the geometric relation 
\[
Kp\subseteq p\subseteq\Diamond_{K}p.
\]
The region $Kp$ contains states where flood-criticality is robust.
The region $p\wedge\neg Kp$ contains states where flooding occurs
but is not robust. The region $\Diamond_{K}p\setminus p$ contains
states where flooding does not occur at the actual state but remains
live nearby. The Moorean region is 
\[
\{(x,y)\in W:(p\wedge\neg Kp)(x,y)=1\}.
\]
It consists of flood states that are not robust enough to be known
under the evidence relation. Geometrically, it is the band inside
$p$ but outside $Kp$: the actual-flood-but-not-known region.

\begin{figure}
\begin{centering}
\includegraphics[scale=0.5]{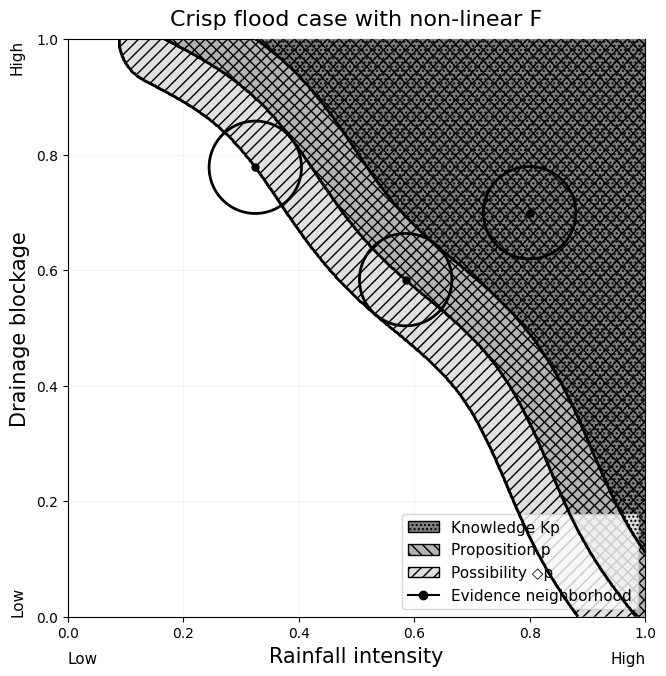}
\par\end{centering}
\caption{A two-dimensional crisp flood-risk example with a non-linear flood-stress
function $F$. The horizontal axis represents rainfall intensity and
the vertical axis represents drainage blockage. The proposition $p$
is the set of states where $F(x,y)\ge c$, i.e. the flood-critical
region. The knowledge region $Kp$ is the robust interior of $p$:
at every point in $Kp$, the entire evidence neighbourhood remains
inside $p$. The possibility region $\Diamond p$ is the outer expansion
of $p$: at every point in $\Diamond p$, the evidence neighbourhood
intersects $p$. The circles illustrate evidence neighbourhoods. A
circle fully contained in $p$ corresponds to knowledge-like support;
a circle centred in $p$ but crossing outside $p$ lies in the Moorean
region $p\wedge\neg Kp$; and a circle centred outside $p$ but intersecting
$p$ represents live possibility without actual flood-criticality.}
\end{figure}

A point deep inside $Kp$ represents a state where all nearby admissible
states are also flood states. A point in $p\wedge\neg Kp$ represents
a state where flooding occurs, but some nearby admissible states leave
the flood region. Finally, a point in $\Diamond_{K}p\setminus p$
represents a state where flooding does not occur at the actual point,
but remains live in a nearby admissible state.

\begin{figure}
\begin{centering}
\includegraphics[scale=0.5]{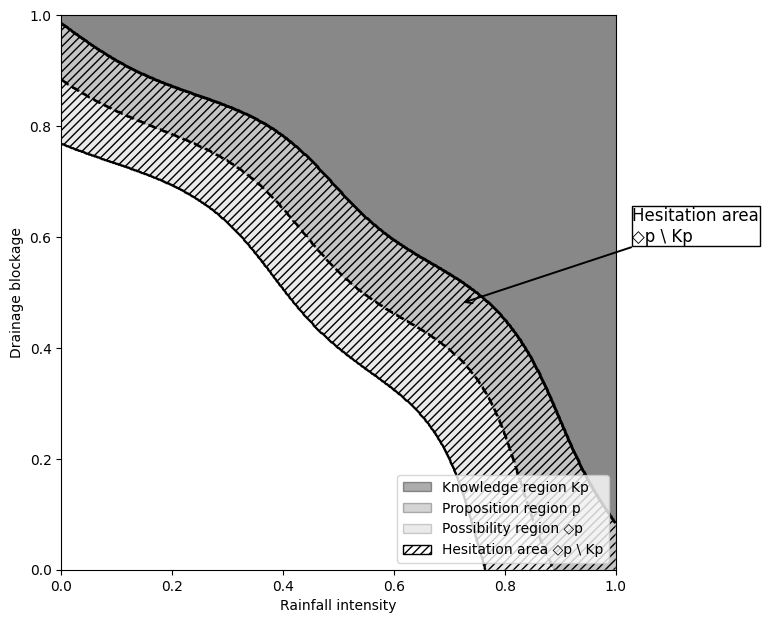}
\par\end{centering}
\caption{Depicting the hesitation area for simplified flood example.}
\end{figure}

The flood-risk example is motivated by hydrological and hydraulic
modelling, where flood risk depends on interacting variables such
as rainfall, saturation, runoff, blockage, freeboard, tide, and climate
state \cite{ChowMaidmentMays1988Applied,Beven2012Rainfall,HallDawsonSayers2003Integrated,BatesEtAl2010LISFLOOD}.
The crisp tolerance-ball formulation is a simple version of the same
idea used in fuzzy evidence models: nearby hydrological or infrastructural
states are treated as relevant alternatives for assessing robustness
and live possibility.

The hesitation region is 
\[
T(\Diamond_{K}p)\setminus T(Kp).
\]
It consists of states where the institution cannot make a fully robust
classification. Some of these states are actually flood-critical but
not robustly so. Others are not flood-critical but close enough to
the flood boundary that the risk cannot responsibly be ruled out.

This example gives a concrete interpretation of the epistemic operators:
\[
p:\quad\text{the state is flood-critical;}
\]
\[
Kp:\quad\text{the flood-critical classification is robust;}
\]
\[
\Diamond_{K}p:\quad\text{flood-criticality remains live nearby;}
\]
\[
p\wedge\neg Kp:\quad\text{the state is flood-critical but epistemically fragile;}
\]
\[
\Diamond_{K}p\wedge\neg Kp:\quad\text{flood-criticality lies in the hesitation region.}
\]
The practical meaning is direct. A point in $Kp$ supports formal
endorsement of the flood-critical classification. A point in $p\wedge\neg Kp$
supports sensitivity analysis, drainage inspection, or temporary defensive
action. A point in $\Diamond_{K}p\setminus p$ supports monitoring,
because flooding is not actual at the point but remains live under
admissible perturbation. A point outside $\Diamond_{K}p$ may be closed,
subject to audit justification.

The size of the hesitation region depends on $\beta$. A larger $\beta$
imposes a stricter robustness standard and widens the hesitation region.
A smaller $\beta$ narrows the evidence neighbourhood and makes robust
endorsement easier. The choice of $\beta$ is therefore not a technical
detail alone. It encodes the institution's tolerance for perturbation
when classifying flood risk.

\section{Probability and modal status}

\label{sec:Probability-and-modal}

The semantics developed here is modal rather than probabilistic. This
does not displace probabilistic or Bayesian methods, which remain
indispensable in QRM. The point is more precise: probability and modal
status answer different governance questions about the same risk proposition. 

\subsection{Probabilistic readings and evidential structure}

Probability asks how likely $p$ is, or how much of the relevant state
space supports $p$. The modal-evidential framework asks whether $p$
is endorsed, live, excluded, fragile, or unresolved under the institution's
evidence relation. Neither question subsumes the other.

Let $p\in\mathrm{Risk}(S)$. At a state $w$, the modal-evidential
framework distinguishes three governance statuses.
\begin{itemize}
\item First, $p$ may be endorsed: $Mp(w)$ is high because $p$ is supported
throughout the evidence field around $w$.
\item Second, $p$ may be live but not endorsed: $\Diamond_{M}p(w)$ is
high while $Mp(w)$ is low.
\item Third, $p$ may be excluded: $\Diamond_{M}p(w)$ is low because no
relevant alternative supports $p$ to a significant degree.
\end{itemize}
These statuses depend on which alternatives matter, not only on how
many alternatives support $p$. A probability measure $\mu$ aggregates
over states and returns a scalar. Two propositions may therefore have
the same probability while occupying different structural positions
in the evidence field.
\begin{example}[Equal probability, divergent modal status]
Let 
\[
W=\{w,u_{1},u_{2},u_{3},u_{4}\},
\]
with uniform measure $\mu(v)=1/5$ for every $v\in W$. Let the crisp
evidence set at $w$ be 
\[
\Gamma_{K}(w)=\{w,u_{1},u_{2}\}.
\]
Define 
\[
T(p)=\{w,u_{1},u_{2}\},\qquad T(q)=\{w,u_{3},u_{4}\}.
\]
Then 
\[
\mu(T(p))=3/5=\mu(T(q)).
\]
However, 
\[
Kp(w)=1,\qquad Kq(w)=0,
\]
because $\Gamma_{K}(w)\subseteq T(p)$, while $\Gamma_{K}(w)\not\subseteq T(q)$.
Moreover, 
\[
\Diamond_{K}q(w)=1,
\]
since $w\in\Gamma_{K}(w)\cap T(q)$. Thus $p$ is robustly endorsed
and may support sign-off, while $q$ is merely live and calls for
further investigation, despite the two propositions having identical
probability. 
\end{example}

The same point applies when probability is high. Let $|W|=10$, let
$\mu$ be uniform, and let 
\[
T(r)=W\setminus\{v_{2}\}.
\]
Then $\mu(T(r))=0.9$. If the evidence set at $w$ includes $v_{2}$,
then 
\[
\mu(T(r))=0.9,\qquad Kr(w)=0.
\]
A model may therefore produce favourable estimates with high probability
across a broad prior while remaining fragile under a specific nearby
alternative specification. Probability alone does not reveal this
local fragility; the modal operator does.

The Moorean expression $p\wedge\neg Mp$ has a natural probabilistic
reading: 
\[
\mathbb{P}(p\wedge\neg Mp)=\mu\bigl(T(p)\setminus T(Mp)\bigr).
\]
This measures the prevalence of states in which $p$ holds without
endorsement. It is a useful aggregate statistic, but it is not the
governance meaning of the diagnostic.

At a state $w$, the crisp diagnostic says 
\[
w\models p\wedge\neg Kp\quad\Longleftrightarrow\quad w\in T(p)\text{ and }\Gamma_{K}(w)\not\subseteq T(p).
\]
Thus $p$ is true at the actual world, but the evidence neighbourhood
extends beyond the truth set of $p$. In risk language, the risk obtains,
but the claim is not stable across the institution's admissible alternatives.

This is a claim about the geometry of the evidence field relative
to the truth boundary of $p$, not merely about the measure of a subset
of $W$. It cannot be recovered from $\mu(T(p))$, from $\mu(T(Mp))$,
or from any scalar function of $\mu$ alone. The same aggregate prevalence
of $T(p)\setminus T(Mp)$ may arise from a nearby model alternative,
a remote tail scenario, or a systematic gap in the evidence architecture.

To connect the two perspectives, let $\mu$ be a probability measure
on $W$, and let $p:W\to[0,1]$ be a graded proposition. In the crisp
case, $p=\mathbf{1}_{A}$ for $A:=\{w\in W:p(w)=1\}$, and 
\[
\mathbb{P}(p):=\mu(A)=\int_{W}p(w)\,d\mu(w).
\]
In the graded case, $p$ may be treated as a fuzzy event, with probability
given by expectation: 
\[
\mathbb{P}(p):=\int_{W}p(w)\,d\mu(w).
\]
This generalises the crisp case because 
\[
\int_{W}\mathbf{1}_{A}\,d\mu=\mu(A).
\]
Alternatively, for any threshold $\eta\in[0,1]$, one may form the
level event 
\[
A_{\eta}(p):=\{w\in W:p(w)\ge\eta\},
\]
and consider $\mu(A_{\eta}(p))$. The expectation measures average
degree of satisfaction; the level probability measures the probability
that $p$ holds to at least degree $\eta$. This connects the present
framework to standard treatments of fuzzy events and probabilities
\cite{Zadeh1968ProbabilityFuzzyEvents,Zadeh1965FuzzySets}.

A sharper comparison is obtained by restricting probability to the
evidence neighbourhood. For each $w\in W$, let $\mu_{w}$ be a probability
measure supported on $\Gamma_{M}(w)$. Define the local probability
of $p$ at $w$ by 
\[
\rho_{p}(w):=\int_{W}p(v)\,d\mu_{w}(v).
\]
If $p=\mathbf{1}_{A}$, then $\rho_{p}(w)=\mu_{w}(A)$, the conditional
probability of $A$ within the evidence field at $w$.

Under a full-support local measure on $\Gamma_{M}(w)$, modal and
probabilistic statuses are related by limiting cases: 
\[
Mp(w)=1\Longrightarrow\rho_{p}(w)=1,
\]
and 
\[
\Diamond_{M}p(w)=1\Longrightarrow\rho_{p}(w)>0.
\]
The Moorean diagnostic also has a local probabilistic reading. In
the crisp case, 
\[
w\models p\wedge\neg Mp\quad\Longleftrightarrow\quad p(w)=1\text{ and }\rho_{p}(w)<1.
\]
Equivalently, $p(w)=1$ and $\mu_{w}(T(\neg p))>0$: the actual world
satisfies $p$, but the local evidence neighbourhood assigns positive
probability to $\neg p$.

Even so, the modal reading remains stronger for governance purposes.
The condition $Mp(w)=1$ requires every relevant alternative to support
$p$. A high local probability, such as $\rho_{p}(w)=0.9$, does not
yield $Mp(w)=1$. For assurance-grade certification, this difference
is material: admissible counterexamples may block sign-off even when
they barely affect the local average.
\begin{rem}
[Low probability, high modal relevance] The framework does not identify
probability-zero exceptions with impossibility. In a continuous risk
model, a measure-zero set may still contain the nearest stress scenario,
the most fragile model specification, or the most decision-relevant
tail state. This is one reason modal accounts of epistemic risk are
not reducible to probability assignments alone \cite{Pritchard2016EpistemicRisk,Hawthorne2003Knowledge,Babic2019TheoryEpistemicRisk}.
Pritchard's distinction between probability and modal closeness is
especially apt: a low-probability event may be modally close if only
a small perturbation is required to reach it \cite{Pritchard2016EpistemicRisk}. 
\end{rem}

One might try to approximate the infimum and supremum in the definitions
of $M$ and $\Diamond_{M}$ by integration against a local probability
measure $\mu_{w}$. Define 
\[
(M^{\mathrm{agg}}p)(w):=\int_{W}\bigl(\gamma_{M}(w,v)\Rightarrow p(v)\bigr)\,d\mu_{w}(v),
\]
and 
\[
(\Diamond_{M}^{\mathrm{agg}}p)(w):=\int_{W}\bigl(\gamma_{M}(w,v)\otimes p(v)\bigr)\,d\mu_{w}(v).
\]
These are meaningful scores: the first measures average implicational
support, and the second measures average live compatibility with $p$.
They may be useful robustness indices. But they do not replace modal
operators, because they need not satisfy the structural properties
required for knowledge-like or belief-like governance.
\begin{prop}
[Structural failures of aggregated operators] The aggregated operators
$M^{\mathrm{agg}}$ and $\Diamond_{M}^{\mathrm{agg}}$ may fail each
of the following properties, even when the underlying evidence relation
$\gamma_{M}$ is reflexive.
\begin{itemize}
\item \textbf{Factivity.} $M^{\mathrm{agg}}p\le p$ may fail. Averaging
allows support from other worlds to compensate for failure at the
actual world. By contrast, the infimum-based operator satisfies factivity
under reflexivity because the actual-world term cannot be averaged
away: 
\[
(Mp)(w)\le\gamma_{M}(w,w)\Rightarrow p(w)=1\Rightarrow p(w)=p(w).
\]
\item \textbf{Non-exclusion.} $\Diamond_{M}^{\mathrm{agg}}p(w)>0$ may fail
even when $p$ is live at a relevant alternative. If the only relevant
$p$-world has $\mu_{w}$-measure zero, as may occur in continuous
models, then $\Diamond_{M}^{\mathrm{agg}}p(w)=0$ even though $\Diamond_{M}p(w)>0$.
\item \textbf{Conjunction closure.} For a threshold $\eta$, one may have
$\mathbb{P}(p)\ge\eta$ and $\mathbb{P}(q)\ge\eta$ while $\mathbb{P}(p\wedge q)<\eta$.
This is the lottery-style and preface-style problem for high-probability
belief \cite{Kyburg1961Probability,Makinson1965Preface,Foley1992WorkingWithoutNet,Christensen2004PuttingLogic}.
The same failure applies to averaged evidence scores.
\item \textbf{KD45-style conditions.} Consistency, positive introspection,
and negative introspection correspond to seriality, transitivity,
and Euclideanness of the accessibility relation \cite{Chellas1980Modal,BlackburnDeRijkeVenema2001Modal,FaginHalpernMosesVardi1995Reasoning,MeyerVanDerHoek1995Epistemic}.
A probability measure over worlds does not by itself impose these
frame conditions.
\end{itemize}
\end{prop}

\begin{proof}[Proof sketch]
For factivity, let $W=\{w,u\}$, let $\mu_{w}$ be uniform, let $\gamma_{M}(w,w)=\gamma_{M}(w,u)=1$,
and let $p(w)=0$, $p(u)=1$. Using the G\"odel implication, 
\[
(M^{\mathrm{agg}}p)(w)=\frac{1}{2}(0+1)=\frac{1}{2},
\]
but $p(w)=0$. For non-exclusion, in a continuous state space let
$v^{*}$ be the unique $p$-supporting state in $\Gamma_{M}(w)$,
with $\mu_{w}(\{v^{*}\})=0$. Then $\Diamond_{M}p(w)>0$, but $\Diamond_{M}^{\mathrm{agg}}p(w)=0$.
For conjunction closure, standard lottery constructions apply. Finally,
KD45-style properties are structural properties of the accessibility
relation, not consequences of probability measures. 
\end{proof}
Aggregated operators are therefore useful as scores, such as average
robustness, expected local support, or weighted live compatibility.
But they are not substitutes for modal operators, because they lack
the structural properties of factivity, non-exclusion, closure, and
introspection on which the governance architecture depends.

\subsection{Probability and modal status jointly}

Probability and modal status are complementary. Probability answers
prevalence questions: 
\[
\mathbb{P}(p):\text{ how likely is the risk?}
\]
\[
\mathbb{P}(p\wedge\neg Mp):\text{ how prevalent is the Moorean region?}
\]
\[
\rho_{p}(w):\text{ what is the local probability within the evidence neighbourhood?}
\]
Modal status answers structural questions: 
\[
Mp(w):\text{ is the risk robustly endorsed across all relevant alternatives?}
\]
\[
\Diamond_{M}p(w):\text{ is the risk live somewhere in the evidence field?}
\]
\[
p(w)\wedge\neg Mp(w):\text{ is the risk present but not robust at this state?}
\]
\[
H_{M}(p)(w):\text{ what is the gap between what can be endorsed and what cannot be excluded?}
\]

The two dimensions cross-classify risk states into four governance-relevant
quadrants: 
\[
\begin{array}{c|cc}
 & \rho_{p}(w)\text{ high} & \rho_{p}(w)\text{ low}\\
\hline Mp(w)\text{ high} & \text{robust, likely} & \text{robust locally, unlikely}\\
Mp(w)\text{ low} & \text{likely but fragile} & \text{unlikely and unsupported}
\end{array}
\]

The first quadrant, robust and likely, describes a claim supported
both structurally and probabilistically. It can bear assurance-grade
endorsement: sign-off, regulatory disclosure, and binding controls.

The second quadrant, robust locally but unlikely, describes a structurally
stable tail scenario. It may occur rarely, but if severity is high,
it may still justify capital allocation or contingency planning.

The third quadrant, likely but fragile, is the most governance-critical.
The risk has high aggregate or local probability, but the claim is
not robust under relevant perturbations. A model may report favourable
tail estimates with high frequency while a nearby alternative model
specification, stress calibration, or data shift breaks the conclusion.
This quadrant contains false reassurance: probability looks comforting,
while modal structure reveals fragility. Governance responses then
is warranted.

The fourth quadrant, unlikely and unsupported, describes a risk with
low probability and no robust endorsement. Non-exclusion may still
matter: if $\Diamond_{M}p>0$ and severity is high, the item may require
monitoring or precaution. If $\Diamond_{M}p$ is also low, the item
may be closed with audit justification.

Thus probability alone cannot flag the likely-but-fragile quadrant.
A high value of $\mathbb{P}(p)$ or $\rho_{p}(w)$ may provide false
assurance when the evidence neighbourhood is inhomogeneous. The modal
operator $Mp$ is designed to detect precisely this inhomogeneity.
This is why $p\wedge\neg Mp$ is a governance tool, not merely a logical
curiosity.
\begin{rem}
[Fragility in model risk] The likely-but-fragile quadrant corresponds
to a common failure mode in model risk governance. A model may pass
standard back-testing and validation checks while remaining sensitive
to small changes in model class, calibration window, or data vintage.
SR 11-7 model-risk guidance \cite{FederalReserveOCC2011SR117} treats
development, validation, implementation, and use as distinct control
points precisely because aggregate performance may conceal local fragility.
The modal operator $Mp$ formalises this concern: it asks not whether
the model usually works, but whether the conclusion is stable under
the admissible perturbations encoded by the evidence relation. 
\end{rem}

\begin{rem}
[The role of severity] Governance responses in the second and fourth
quadrants are severity-dependent. The modal framework captures this
through non-exclusion: $\Diamond_{M}p(w)>0$ means that $p$ remains
live somewhere in the evidence field, regardless of its aggregate
probability. For high-severity risks, such as systemic contagion,
catastrophic flood, or critical infrastructure failure, even weak
non-exclusion may justify precautionary action. This is the formal
analogue of the precautionary principle: action may be warranted by
the inability to rule out a severe outcome, even without high-probability
endorsement \cite{Sunstein2005Laws,Stirling2007Risk,Hansson2004Philosophical}. 
\end{rem}

\subsubsection{Probabilistic reading of the flood example}

\label{sec:flood:prob}

The simplified flood model gives a concrete geometric realisation
of that classification. Letting $Z=(X,Y)$ be a random vector on $W=[0,1]^{2}$,
where $X$ denotes rainfall intensity and $Y$ denotes drainage blockage,
we denote the probability by $\mu$. The flood event is 
\[
A:=\{(x,y)\in W:F(x,y)\ge c\}.
\]
Equivalently, with $p(w)=\mathbf{1}_{A}(w)$, 
\[
\mathbb{P}(A)=\int_{W}p(w)\,d\mu(w)=\int_{W}\mathbf{1}_{\{F(w)\ge c\}}\,d\mu(w).
\]
If $\mu$ has density $f_{X,Y}$, this becomes 
\[
\mathbb{P}(A)=\int_{0}^{1}\int_{0}^{1}\mathbf{1}_{\{F(x,y)\ge c\}}f_{X,Y}(x,y)\,dx\,dy.
\]
In the special case where $\mu$ is uniform on $W=[0,1]^{2}$, $\mathbb{P}(A)=\lambda_{2}(A)$,
where $\lambda_{2}$ denotes two-dimensional Lebesgue measure.

This probabilistic quantity answers the first-order question: \emph{how
likely is flooding?} The modal operators answer a different question,
and the two together generate the governance quadrants described earlier.

For each state $w$, let $\mu_{w}$ be a probability measure supported
on $\Gamma_{K}(w)$. For concreteness, take $\mu_{w}$ to be the uniform
distribution on $\Gamma_{K}(w)\cap W$. Define the local flood probability
\[
\rho_{A}(w):=\mu_{w}(A)=\int_{\Gamma_{K}(w)}\mathbf{1}_{A}(v)\,d\mu_{w}(v).
\]
In the crisp setting, the modal statuses and local probability are
related by limiting cases: 
\[
Kp(w)=1\;\Longrightarrow\;\rho_{A}(w)=1,\qquad\text{and}\qquad\Diamond_{K}p(w)=1\;\Longleftrightarrow\;\rho_{A}(w)>0,
\]
provided $\mu_{w}$ assigns positive mass to every relevant open part
of $\Gamma_{K}(w)$. The key point is that the converse of the first
implication fails: $\rho_{A}(w)$ can be high without reaching $1$,
and in that case $Kp(w)=0$.

\paragraph{The four quadrants in the flood plane.}

The two-dimensional flood model partitions $W=[0,1]^{2}$ into four
governance-relevant regions, corresponding to the quadrants:
\begin{enumerate}
\item \emph{Robust and likely} ($Kp(w)=1$, $\rho_{A}(w)=1$). Points deep
inside the flood region, far from the boundary of~$A$. The entire
evidence ball lies within~$A$. Flood-criticality is both certain
within the evidence neighbourhood and supported by high aggregate
probability. Governance response: endorse the flood classification,
activate binding controls, report for regulatory purposes.
\item \emph{Robust locally but unlikely} ($Kp(w)=1$, $\mathbb{P}(A)$ low).
This quadrant is realised when the flood region $A$ is small in area
(low global probability), but the state $w$ lies deep inside $A$
so that $\Gamma_{K}(w)\subseteq A$. Flood-criticality is structurally
stable at $w$ even though $A$ is globally rare. Governance response:
severity-dependent. If the flood consequence is severe, the structural
robustness at~$w$ justifies capital allocation and contingency planning
despite low global frequency.
\item \emph{Likely but fragile} ($Kp(w)=0$, $\rho_{A}(w)$ high). Points
inside $A$ but near the boundary, where most of the evidence ball
lies within $A$ but a sliver extends outside. Flood-criticality holds
at the actual state, and most nearby states agree, but the claim is
not robust: a small perturbation in rainfall, drainage, or model calibration
could remove the state from the flood region. This is the Moorean
band $p\wedge\neg Kp$ with high local probability. This is the most
governance-critical quadrant, because the high local probability may
give a misleading impression of safety. Governance response: sensitivity
analysis on the flood-stress function and drainage inspection. 
\item \emph{Unlikely and unsupported} ($Kp(w)=0$, $\rho_{A}(w)$ low).
Points outside or barely inside $A$ with few flood states in the
evidence neighbourhood. Non-exclusion ($\Diamond_{K}p(w)=1$) may
still hold if any part of $\Gamma_{K}(w)$ intersects $A$. Governance
response: if $\Diamond_{K}p=0$, the flood item may be closed with
justification. If $\Diamond_{K}p=1$, the flood risk remains live
and requires monitoring, especially under high-severity conditions. 
\end{enumerate}

\paragraph{A concrete boundary state.}

To make the third quadrant vivid, consider a specific point $w_{0}=(x_{0},y_{0})$
lying just inside the boundary of $A$, so that $F(x_{0},y_{0})=c+\varepsilon$
for a small $\varepsilon>0$. The evidence ball $\Gamma_{K}(w_{0})$
straddles the boundary: most of its area lies within $A$, but a fraction
extends into the non-flood region $A^{c}$. Suppose the geometry gives
$\rho_{A}(w_{0})=0.92.$ That is, $92\%$ of the evidence neighbourhood
is flood-critical. A probabilistic assessment might treat this as
strong evidence of flood risk. But $Kp(w_{0})=0$, because the remaining
$8\%$ of the evidence ball lies outside $A$. The flood classification
at $w_{0}$ is fragile: a small shift in rainfall intensity or drainage
condition could move the state across the boundary.

This is precisely the situation that the modal operator is designed
to detect. The local probability $\rho_{A}(w_{0})=0.92$ aggregates
over the evidence neighbourhood and returns a scalar. The modal status
$Kp(w_{0})=0$ records that the evidence neighbourhood is not uniformly
supportive. For assurance-grade governance, the latter is the more
important diagnostic: it identifies a state where the flood classification
could easily have been different under a nearby admissible alternative.

\paragraph{Geometric decomposition of the hesitation region.}

The hesitation region $T(\Diamond_{K}p)\setminus T(Kp)$ can now be
decomposed along both the modal and probabilistic dimensions: 
\begin{itemize}
\item The \emph{inner hesitation band} is the Moorean region $p\wedge\neg Kp$.
Here flooding is actual but not robust. Local probability satisfies
$0<\rho_{A}(w)<1$. Points near the interior of $A$ have $\rho_{A}$
close to~$1$ (likely but fragile); points near the boundary of $A$
from the inside have $\rho_{A}$ closer to the area fraction of the
evidence ball that intersects $A$.
\item The \emph{outer hesitation band} is the possibility-only region $\Diamond_{K}p\setminus p$.
Here flooding is not actual but remains live. Local probability satisfies
$0<\rho_{A}(w)<1$ as well, but now $p(w)=0$. These are states where
the institution cannot responsibly exclude flood risk, even though
the actual state is not flood-critical. 
\end{itemize}
The width of each band is controlled by the tolerance radius $\beta$.
A larger $\beta$ (wider evidence neighbourhood, more conservative
robustness standard) produces a wider hesitation region and a narrower
knowledge region. A smaller $\beta$ (tighter evidence standard) shrinks
hesitation and expands knowledge. The choice of $\beta$ is itself
a governance decision: it encodes how much perturbation the institution
considers admissible when certifying a flood classification.

\paragraph{Why the joint picture matters for flood governance.}

In practice, flood governance often begins with the global probability
$\mathbb{P}(A)$ and proceeds to local assessments of return periods,
peak discharges, and scenario likelihoods. The modal framework does
not displace this practice. It adds a structural layer that answers
a question probability leaves open: \emph{is the flood classification
at a given state robust under the admissible perturbations encoded
by the evidence relation?}

The general argument for why this structural layer is needed, and
why it is not reducible to a probability threshold. The flood model
instantiates that argument geometrically. The Moorean band $p\wedge\neg Kp$
is the set of states where the flood is real but fragile. The hesitation
region $T(\Diamond_{K}p)\setminus T(Kp)$ is the set of states where
the institution lacks a decisive epistemic stance. The local probability
$\rho_{A}(w)$ quantifies the degree of local exposure. Together,
these three quantities, modal status, hesitation width, and local
probability, give a richer governance picture than any one of them
alone.

The cross-classification provides the governance rule: 
\begin{itemize}
\item $Kp=1$ and $\rho_{A}$ high: endorse the flood classification. 
\item $Kp=0$ and $\rho_{A}$ high (likely but fragile): escalate, inspect,
deploy temporary defences. 
\item $Kp=0$, $\Diamond_{K}p=1$, and $\rho_{A}$ low: monitor, especially
under high-severity conditions. 
\item $\Diamond_{K}p=0$: close the flood item with audit justification. 
\end{itemize}
The second case, \emph{likely but fragile}, is the one that probability
alone cannot flag as governance-critical, and it is precisely the
case that the Moorean diagnostic $p\wedge\neg Kp$ is designed to
identify. 

\section{Meta-Level Controls as an Epistemic Architecture}

\label{sec:meta}

The Moorean limitation shows that epistemic diagnostics cannot always
be handled as ordinary object-level risk claims. A diagnostic such
as 
\[
p\wedge\neg Kp
\]
says that $p$ obtains while the institution lacks assurance-grade
endorsement of $p$. If the same framework then requires the institution
to know this diagnostic in the ordinary object-level sense, the diagnostic
becomes unstable: knowing $p\wedge\neg Kp$ pushes the institution
toward knowing $p$. The same issue arises for 
\[
p\wedge\neg Bp,
\]
where adopting the diagnostic as an ordinary working commitment threatens
the claim that $p$ is not yet a working commitment.

The response is architectural. Risk governance should separate object-level
risk claims from meta-level epistemic diagnostics. Object-level claims
concern hazards, exposures, controls, losses, and adverse states of
the world. Meta-level diagnostics concern the institution's relation
to those claims: whether they are known, believed, modelled, validated,
excluded, or left unresolved.

The core distinction is therefore 
\[
\text{object-level risk}\qquad\text{versus}\qquad\text{meta-level epistemic status}.
\]
A mature risk programme needs both levels. It must ask whether $p$
obtains, but also whether $p$ is represented in the model architecture,
visible to decision-makers, validated, excluded for defensible reasons,
and assigned to a responsible owner. 

\subsection{Object-level and meta-level risk}

Let 
\[
\Risk_{0}(S)
\]
be the object-level risk family. Its elements are first-order propositions
about the subject $S$. Examples include: 
\[
\text{tail loss exceeds the capital threshold;}
\]
\[
\text{a systemic cascade affects a material fraction of banks;}
\]
\[
\text{a protected flood zone is in a critical overtopping state;}
\]
\[
\text{a control fails under stress;}
\]
\[
\text{a hedge is insufficient under the relevant scenario.}
\]
These propositions concern the world, system, portfolio, network,
model, or infrastructure.

Let 
\[
\Risk_{1}(S)
\]
be the meta-level risk family. Its elements are propositions about
the institution's epistemic relation to object-level claims. Examples
include 
\[
p\wedge\neg Kp,
\]
\[
p\wedge\neg Bp,
\]
\[
\Diamond_{M}p\wedge\neg Mp,
\]
\[
\overline{M}p\wedge\neg Mp,
\]
and 
\[
H_{M}(p)>\eta.
\]

Each has a different governance interpretation. The diagnostic $p\wedge\neg Kp$
says that $p$ obtains but lacks assurance-grade endorsement. This
is a certification gap. The diagnostic $p\wedge\neg Bp$ says that
$p$ obtains but has not been adopted as a working commitment. This
is an action gap. The diagnostic $\Diamond_{M}p\wedge\neg Mp$ says
that $p$ remains live but is not positively supported enough for
endorsement. This is a hesitation or non-exclusion gap. The condition
$H_{M}(p)>\eta$ says that the gap between what the institution can
endorse and what it cannot rule out exceeds the governance tolerance.
This is a hesitation-margin breach.

The key point is that these are not first-order hazards. They are
diagnostics of how the institution is positioned relative to first-order
hazards. They should therefore be governed differently.

\subsection{The audit operator}

Introduce an audit operator 
\[
A:\mathcal{L}\to\mathcal{L}.
\]
The intended reading is 
\[
Ad=\text{diagnostic }d\text{ is audit-visible and governable.}
\]
The operator $A$ should not be identified with $K$ or $B$. These
operators have different functions: 
\[
Kp:\quad\text{certify }p\text{ for assurance-grade use;}
\]
\[
Bp:\quad\text{adopt }p\text{ as a disciplined working commitment;}
\]
\[
Ad:\quad\text{record and govern diagnostic }d.
\]

For example, 
\[
A(p\wedge\neg Kp)
\]
does not say that the institution knows $p\wedge\neg Kp$ as an ordinary
object-level claim. It says that the institution has made the diagnostic
visible to audit or review.

Similarly, 
\[
A(\Diamond_{M}p\wedge\neg Mp)
\]
records that $p$ remains live without endorsement. The point is not
to certify $p$ immediately. The point is to prevent the live-but-unendorsed
status from disappearing from governance view.

A minimal requirement on $A$ is monotonicity: 
\[
d\le e\Rightarrow Ad\le Ae.
\]
If one diagnostic entails another, auditability should not decrease.

A second requirement is persistence: 
\[
Ad\le AAd.
\]
If a diagnostic is audit-visible, then its audit record should itself
remain visible. This prevents unresolved epistemic gaps from being
lost through organisational memory failures.

Unlike $K$, the audit operator should not generally be factive. Auditing
a diagnostic does not certify the object-level truth of the risk claim.
It records that the diagnostic is sufficiently important to be governed.
This distinction is crucial. Many risk-management processes must track
concerns that are unresolved, disputed, partially supported, outside
model scope, or live under stress scenarios.

\subsection{Typed reachability}

The collapse pressure arises from applying one reachability principle
to everything. The correction is to type reachability by level.

For object-level risks, use the object-level reach principle: 
\[
\mathrm{Reach}_{0}:\qquad p\in\Risk_{0}(S)\Rightarrow p\le\Diamond_{M}Mp.
\]
This says that if $p$ is a real and decision-relevant object-level
risk, then an appropriate institutional stance toward $p$ should
be reachable. The stance may be assurance-grade endorsement, working
commitment, or another domain-specific status.

For meta-level diagnostics, use a different principle: 
\[
\mathrm{Reach}_{1}:\qquad d\in\Risk_{1}(S)\Rightarrow d\le\Diamond_{A}A(d).
\]
This says that if $d$ is a meta-level diagnostic, then auditability
of $d$ should be reachable.

Now let 
\[
d=p\wedge\neg Mp.
\]
The untyped framework pushes toward 
\[
M(p\wedge\neg Mp),
\]
which creates Moorean pressure. The typed framework instead requires
\[
A(p\wedge\neg Mp).
\]
This is the central architectural move. A diagnostic of missing endorsement
should be governed by audit, not forced into the endorsement register
whose absence it records.

The typed structure is therefore: 
\[
p\in\Risk_{0}(S)\Rightarrow p\le\Diamond_{M}Mp,
\]
but 
\[
p\wedge\neg Mp\in\Risk_{1}(S)\Rightarrow p\wedge\neg Mp\le\Diamond_{A}A(p\wedge\neg Mp).
\]
This preserves the governance ambition of reachability while avoiding
collapse.

\subsection{A governance rule for meta-level diagnostics}

For each object-level risk $p$ and standard $M\in\{K,B\}$, the meta-level
layer should record: 
\[
Mp:\quad\text{Is \ensuremath{p} positively supported or endorsed?}
\]
\[
\Diamond_{M}p:\quad\text{Does \ensuremath{p} remain live under the evidence relation?}
\]
\[
\overline{M}p:\quad\text{Has \ensuremath{p} not been ruled out?}
\]
\[
H_{M}(p):\quad\text{How large is the gap between endorsement and non-exclusion?}
\]
\[
\mathcal{R}_{M}(p):\quad\text{Does \ensuremath{p} obtain without the relevant institutional stance?}
\]
A governance rule can then be written as 
\[
G:\bigl(p,Mp,\Diamond_{M}p,\overline{M}p,H_{M}(p),\mathcal{R}_{M}(p)\bigr)\mapsto\mathcal{A},
\]
where $\mathcal{A}$ is a set of governance actions. A simple threshold
version is: 
\[
Mp\ge\alpha\Rightarrow\text{endorse }p,
\]
\[
\Diamond_{M}p\ge\beta\text{ and }Mp<\alpha\Rightarrow\text{monitor or escalate }p,
\]
\[
H_{M}(p)\ge\eta\Rightarrow\text{record hesitation and require review},
\]
\[
\mathcal{R}_{M}(p)\ge\delta\Rightarrow\text{open a meta-level audit item}.
\]
The thresholds $\alpha,\beta,\eta,\delta$ are domain-specific. In
high-severity settings, weak non-exclusion may be enough to trigger
monitoring. In lower-severity settings, escalation may require stronger
support.

The important point is that the response to a diagnostic is not uniform.
A live but unsupported risk may require monitoring. A high-hesitation
risk may require review. A true-but-unendorsed risk may require audit.
A strongly supported risk may require endorsement and action. The
modal framework separates these cases.

\subsection{An epistemic package }

The proposed architecture has four layers.

First, the object-level layer contains ordinary risk claims: 
\[
p\in\Risk_{0}(S).
\]
Second, the epistemic-status layer records endorsement, working commitment,
live possibility, non-exclusion, hesitation, and inconsistency: 
\[
Kp,\quad Bp,\quad\Diamond_{M}p,\quad\overline{M}p,\quad H_{M}(p),\quad I_{M}(p).
\]
Third, the diagnostic layer records epistemic defects: 
\[
p\wedge\neg Kp,\qquad p\wedge\neg Bp,\qquad\Diamond_{M}p\wedge\neg Mp,\qquad H_{M}(p)>\eta.
\]
Fourth, the audit layer governs those diagnostics: 
\[
A(p\wedge\neg Kp),\qquad A(p\wedge\neg Bp),\qquad A(\Diamond_{M}p\wedge\neg Mp),\qquad A(H_{M}(p)>\eta).
\]

This typed architecture preserves the practical ambition of risk governance.
Real risks should be brought within the reach of responsible epistemic
processes. But epistemic gaps should not be collapsed into the same
object-level endorsement register that failed to capture them. They
should be made audit-visible, assigned, reviewed, revised, escalated,
or closed with justification.

The final governance message is therefore: 
\[
\text{object-level risks require endorsement or commitment;}
\]
\[
\text{meta-level epistemic gaps require audit and control.}
\]
This is how the framework preserves epistemic diagnostics without
producing Moorean collapse.

\subsection{Belief revision, update, and commitment management}

Working commitments change over time. A risk may be live at one date,
dismissed at another, rediscovered later, and eventually treated as
if it had never been visible. Meta-level controls therefore require
institutional memory and disciplined revision.

Let $B_{t}p$ mean that $p$ is a working commitment at time $t$.
Let $E_{t}$ denote the evidence state at time $t$. There are two
different kinds of change. First, there is revision. Revision occurs
when new evidence conflicts with the current commitment register.
For example, a model committee may have adopted $Bp$, where $p$
says that a model is sufficiently conservative. If challenger validation
supports $\neg p$, the institution must revise its working commitments.

Second, there is update. Update occurs when the world or risk environment
changes. The prior commitment may have been reasonable at the earlier
time, but the environment has shifted. A market regime may change,
a flood defence may deteriorate, a banking network may acquire new
exposures, or a regulatory rule may alter the relevant threshold.

Write revision as 
\[
B_{t}*e,
\]
where $e$ is new evidence. Write update as 
\[
B_{t}\diamond c,
\]
where $c$ records a change in the environment. For each $p\in\Risk_{0}(S)$,
the meta-level layer records 
\[
B_{t}p,\qquad\Diamond_{B,t}p,\qquad\overline{B}_{t}p,\qquad H_{B,t}(p),
\]
\[
D_{B,t}(p):=p\wedge\neg B_{t}p,
\]
and 
\[
I_{B,t}(p):=B_{t}p\wedge B_{t}\neg p.
\]

A belief-management rule has the form 
\[
U_{B}:\bigl(B_{t}p,\Diamond_{B,t}p,\overline{B}_{t}p,H_{B,t}(p),D_{B,t}(p),I_{B,t}(p),E_{t},c_{t}\bigr)\mapsto B_{t+1}p.
\]
It may be decomposed as 
\[
U_{B}=(R_{B},\mathcal{C}_{B}),
\]
where $R_{B}$ is a revision rule and $\mathcal{C}_{B}$ is an update
rule.

For example: 
\[
I_{B,t}(p)\ge\iota\Rightarrow\text{revise the working-commitment;}
\]
\[
H_{B,t}(p)\ge\eta\Rightarrow\text{require review;}
\]
\[
\Diamond_{B,t}p\ge\beta\text{ and }B_{t}p<\alpha\Rightarrow\text{adopt a precautionary commitment.}
\]
This gives a more precise interpretation of $B$. A $B$-claim is
not a permanent belief. It is an institutional commitment embedded
in revision or update.

Dynamic epistemic logic provides a related perspective. Public announcement
logic and action-model approaches study how epistemic states change
when information is announced or when epistemic actions occur \cite{Plaza1989Public,Gerbrandy1999Bisimulations,BaltagMossSolecki1998Logic,vanDitmarschHoekKooi2008Dynamic,vanBenthem2007Dynamic}.
Organisational learning theory emphasises mechanisms for retaining,
revising, and acting on prior information \cite{ArgyrisSchon1978Organizational,LevittMarch1988Organizational}.
Work on organisational accidents and high-reliability organisations
similarly emphasises sensitivity to weak signals, near misses, latent
conditions, and unresolved vulnerabilities \cite{Reason1997Managing,WeickSutcliffe2007Managing}.

\subsection{Fragmented institutions and multiple operators}

So far, $K$ and $B$ have been treated as single institutional operators.
Real organisations are more fragmented. Different units may use different
data, models, thresholds, incentives, and responsibilities.

Introduce belief-like operators $B_{1},\ldots,B_{n},$ where $B_{i}$
represents the working commitment of unit $i$. Introduce knowledge-like
operators $K_{1},\ldots,K_{m},$ where $K_{j}$ represents assurance
under standard $j$.

For example: 
\[
B_{\mathrm{model}}p:\quad\text{the model owner is committed to }p;
\]
\[
B_{\mathrm{validation}}p:\quad\text{validation is committed to }p;
\]
\[
B_{\mathrm{business}}p:\quad\text{the business line is committed to }p;
\]
\[
K_{\mathrm{audit}}p:\quad\text{audit endorses }p;
\]
\[
K_{\mathrm{reg}}p:\quad\text{the claim has regulatory-grade certification.}
\]

Fragmentation generates cross-operator diagnostics. For example, 
\[
B_{i}p\wedge\neg B_{j}p
\]
records disagreement between units. One unit acts as if $p$, while
another does not.

The diagnostic 
\[
B_{i}K_{j}p\wedge\neg K_{j}p
\]
records mistaken reliance. Unit $i$ acts as if unit $j$ has certified
$p$, but unit $j$ has not.

The diagnostic $B_{i}B_{j}p\wedge\neg B_{j}p$ records a mistaken
higher-order working commitment. Unit $i$ acts as if unit $j$ has
adopted $p$, while unit $j$ has not.

Many real failures are of this type. They concern not just whether
$p$ is known, but who knows $p$, who believes $p$, who owns $p$,
who has ruled out $p$, and who is responsible for escalating $p$.
This connects the architecture to the literature on common knowledge,
distributed information, and interactive epistemology \cite{Aumann1976Agreeing,Geanakoplos1994Common,FaginHalpernMosesVardi1995Reasoning,BattigalliBonanno1999Recent}.

In that literature, agents may interpret signals differently, hold
different expectations, or disagree about fundamentals, and this heterogeneity
can affect prices, trading, bubbles, and risk allocation. Harrison
and Kreps show how heterogeneous expectations can generate speculative
value in asset markets, while Scheinkman and Xiong show how overconfidence
and disagreement can contribute to speculative bubbles \cite{HarrisonKreps1978Speculative,ScheinkmanXiong2003Overconfidence}.

\section{Conclusion}

\label{sec:conclusion}

This paper has argued that quantitative risk management is not only
a matter of identifying, modelling, and measuring adverse states of
the world. It is also a matter of governing the institutional standing
of risk claims. A risk claim may be true without being recognised,
probable without being robust, live without being endorsed, or action-guiding
without being assurance-grade. These distinctions are not captured
by probability alone.

The modal framework developed in the paper represents this second-order
layer explicitly. The formula $Kp$ denotes assurance-grade endorsement
of a risk claim $p$, suitable for certification, regulatory reporting,
board reliance, or audit use. The formula $Bp$ denotes working commitment:
a disciplined action-guiding stance available under incomplete assurance.
The operators $\Diamond_{M}p$ and $\overline{M}p:=\neg M\neg p$
distinguish live reachability from dual non-exclusion. In crisp Boolean
settings these notions coincide; in fuzzy settings they may come apart,
which is why the paper keeps them conceptually separate.

The central diagnostic formulas are 
\[
p\wedge\neg Kp\qquad\text{and}\qquad p\wedge\neg Bp.
\]
They identify cases in which a risk is present but lacks the relevant
institutional stance. The first is an assurance gap: the risk obtains,
but the institution is not entitled to certify or rely on it as known.
The second is an action gap: the risk obtains, but it has not entered
the working-commitment register. More generally, diagnostics such
as 
\[
\Diamond_{M}p\wedge\neg Mp,\qquad\overline{M}p\wedge\neg Mp,\qquad p\wedge M\neg p,
\]
record live-but-unendorsed risk, non-excluded-but-unsupported risk,
and outright institutional conflict.

The main formal lesson is that these diagnostics cannot be treated
as ordinary object-level targets of the same stance whose absence
they record. If $p\wedge\neg Kp$ is forced back into the $K$-register,
the framework generates Moorean and Fitch-style collapse pressure:
knowing the diagnostic pushes toward knowing $p$, undermining the
$\neg Kp$ component. Similarly, if $p\wedge\neg Bp$ is forced back
into the $B$-register, the working-commitment system becomes unstable
under positive introspection and coherence constraints.

The crisp and fuzzy analyses show the same limitation in different
forms. In crisp settings, unrestricted reachability principles push
the framework toward the implausible conclusion that every real risk
is already known, believed, or institutionally endorsed. In fuzzy
settings, the pressure is localised rather than eliminated. Under
factive knowledge-like standards, true-but-unendorsed risk is bounded
by reachable structural uncertainty. Under non-factive belief-like
standards, true-but-uncommitted risk is bounded by reachable epistemic
inconsistency. Grading the semantics therefore gives a more nuanced
picture, but it does not remove the need for architectural separation.

The proposed response is to type the governance architecture. Object-level
risk claims belong to $\Risk_{0}(S)$. Meta-level epistemic diagnostics
belong to $\Risk_{1}(S)$. Object-level claims may be governed by
endorsement, commitment, reachability, validation, and control. Meta-level
diagnostics should instead be made audit-visible through a distinct
audit operator $A$. Thus the appropriate response to a diagnostic
such as $p\wedge\neg Mp$ is not generally 
\[
M(p\wedge\neg Mp),
\]
but rather 
\[
A(p\wedge\neg Mp).
\]
The point is not to certify the diagnostic as an ordinary object-level
risk claim. The point is to record, own, review, escalate, revise,
or close the epistemic gap through a separate governance layer.

This typed architecture preserves the two motivations that initially
appeared to be in tension. It preserves the Risk Management Principle:
real risks without the relevant stance are themselves risk-relevant.
But it also avoids collapse by refusing to treat meta-level diagnostics
as ordinary targets of the very endorsement or commitment register
whose absence they diagnose. Object-level risks require endorsement
or disciplined commitment; meta-level epistemic gaps require audit
and control.

The examples illustrate the practical scope of the framework. In model
risk, $Kp$ separates a model output from assurance-grade robustness
under admissible model variation. In banking-network contagion, $Bp$
explains why a supervisor may act on a stress scenario before the
adverse state is actual. In flood governance, the regions $Kp$, $p\wedge\neg Kp$,
$\Diamond_{K}p\setminus p$, and $\Diamond_{K}p\wedge\neg Kp$ give
a geometric interpretation of robustness, fragility, live possibility,
and hesitation.

The comparison with probability is therefore not a rejection of probabilistic
QRM. Probability remains indispensable for modelling likelihood, tail
loss, exposure, severity, capital impact, and expected loss. The claim
is narrower and structural: probability does not by itself encode
the evidential geometry of endorsement, non-exclusion, robustness,
hesitation, or institutional standing. A risk may be likely but fragile,
unlikely but structurally robust, live but unsupported, or excluded
under one institutional standard while remaining live under another.
Modal status and probability should therefore be used together.

The general conclusion is that epistemic limitation is part of the
risk landscape. Risk governance must govern not only adverse states,
but also the institutional processes by which adverse states become
visible, credible, actionable, certified, challenged, revised, or
closed. The space between risk and endorsement is not empty. It contains
the unknown, the unendorsed, the weakly supported, the non-excluded,
the hesitant, the inconsistent, the fragmented, and the prematurely
closed. That is precisely the space in which meta-level controls matter. 

\bibliographystyle{abbrv}
\bibliography{risk_epistemology_refs}

\section{Proofs for the Moorean Limitation}

\label{app:proofs}

\subsection{Proof of Theorem~\ref{thm:factive-pressure}}

Suppose $q\le p$. By monotonicity of $M$, 
\[
Mq\le Mp.
\]
By factivity, 
\[
Mq\le q.
\]
Since $\wedge$ is a meet, 
\[
Mq\le Mp\wedge q.
\]
By monotonicity of $\Diamond_{M}$, 
\[
\Diamond_{M}Mq\le\Diamond_{M}(Mp\wedge q).
\]

Let $p\in\Risk(S)$. By RMP, 
\[
\mathcal{R}(p)\in\Risk(S).
\]
By RRP, 
\[
\mathcal{R}(p)\le\Diamond_{M}M\mathcal{R}(p).
\]
Since $\mathcal{R}$ is factive as a refinement, 
\[
\mathcal{R}(p)\le p.
\]
Apply the preceding step with $q=\mathcal{R}(p)$. Then 
\[
\Diamond_{M}M\mathcal{R}(p)\le\Diamond_{M}\bigl(Mp\wedge\mathcal{R}(p)\bigr).
\]
Hence 
\[
\mathcal{R}(p)\le\Diamond_{M}\bigl(Mp\wedge\mathcal{R}(p)\bigr).
\]

\subsection{Proof of Corollary~\ref{cor:factive-moore}}

Let 
\[
\mathcal{R}_{M}^{\mathrm{Moore}}(p):=p\wedge\neg Mp.
\]
Then 
\[
p\wedge\neg Mp\le p.
\]
By Theorem~\ref{thm:factive-pressure}, 
\[
p\wedge\neg Mp\le\Diamond_{M}\bigl(Mp\wedge(p\wedge\neg Mp)\bigr).
\]
Since 
\[
Mp\wedge(p\wedge\neg Mp)\le Mp\wedge\neg Mp,
\]
monotonicity gives 
\[
p\wedge\neg Mp\le\Diamond_{M}(Mp\wedge\neg Mp).
\]
Since 
\[
Mp\wedge\neg Mp\le U,
\]
we obtain 
\[
p\wedge\neg Mp\le\Diamond_{M}(Mp\wedge\neg Mp)\le\Diamond_{M}U.
\]

\subsection{Proof of Corollary~\ref{cor:factive-collapse}}

By Corollary~\ref{cor:factive-moore}, 
\[
p\wedge\neg Mp\le\Diamond_{M}U.
\]
If $U=0$, then 
\[
p\wedge\neg Mp\le\Diamond_{M}0.
\]
By bottom preservation, 
\[
\Diamond_{M}0=0.
\]
Thus 
\[
p\wedge\neg Mp\le0.
\]
By conjunction separation, 
\[
p\le Mp.
\]
By factivity, 
\[
Mp\le p.
\]
Therefore 
\[
p=Mp.
\]

\subsection{Proof of Corollary~\ref{cor:factive-conflict}}

Let 
\[
\mathcal{R}_{M}^{\mathrm{anti}}(p):=p\wedge M\neg p.
\]
This is factive as a refinement. By Theorem~\ref{thm:factive-pressure},
\[
p\wedge M\neg p\le\Diamond_{M}\bigl(Mp\wedge(p\wedge M\neg p)\bigr).
\]
Since 
\[
Mp\wedge(p\wedge M\neg p)\le Mp\wedge M\neg p=I_{M}(p),
\]
monotonicity gives 
\[
p\wedge M\neg p\le\Diamond_{M}I_{M}(p).
\]
Since 
\[
I_{M}(p)\le I_{M},
\]
we obtain 
\[
p\wedge M\neg p\le\Diamond_{M}I_{M}(p)\le\Diamond_{M}I_{M}.
\]

\subsection{Proof of Theorem~\ref{thm:belief-internal}}

Let 
\[
q:=p\wedge\neg Mp.
\]
Then 
\[
q\le p,
\]
and 
\[
q\le\neg Mp.
\]
By monotonicity, 
\[
Mq\le Mp,
\]
and 
\[
Mq\le M\neg Mp.
\]
By positive introspection, 
\[
Mp\le MMp.
\]
Hence 
\[
Mq\le MMp,
\]
and 
\[
Mq\le M\neg Mp.
\]
Since $\wedge$ is a meet, 
\[
Mq\le MMp\wedge M\neg Mp.
\]
By definition, 
\[
MMp\wedge M\neg Mp=I_{M}(Mp).
\]
Thus 
\[
M(p\wedge\neg Mp)\le I_{M}(Mp).
\]
Since 
\[
I_{M}(Mp)\le I_{M},
\]
we have 
\[
M(p\wedge\neg Mp)\le I_{M}(Mp)\le I_{M}.
\]

\subsection{Proof of Theorem~\ref{thm:belief-reach}}

Let 
\[
q:=p\wedge\neg Mp.
\]
By RMP, 
\[
q\in\Risk(S).
\]
By RRP, 
\[
q\le\Diamond_{M}Mq.
\]
By Theorem~\ref{thm:belief-internal}, 
\[
Mq\le I_{M}(Mp)\le I_{M}.
\]
By monotonicity of $\Diamond_{M}$, 
\[
\Diamond_{M}Mq\le\Diamond_{M}I_{M}(Mp)\le\Diamond_{M}I_{M}.
\]
Therefore 
\[
p\wedge\neg Mp\le\Diamond_{M}I_{M}(Mp)\le\Diamond_{M}I_{M}.
\]

\subsection{Proof of Corollary~\ref{cor:belief-collapse}}

By Theorem~\ref{thm:belief-reach}, 
\[
p\wedge\neg Mp\le\Diamond_{M}I_{M}.
\]
If $I_{M}=0$, then 
\[
p\wedge\neg Mp\le\Diamond_{M}0.
\]
By bottom preservation, 
\[
\Diamond_{M}0=0.
\]
Hence 
\[
p\wedge\neg Mp\le0.
\]
If conjunction separation holds, then 
\[
p\le Mp.
\]
In the crisp Boolean case, this is 
\[
p\to Mp.
\]

\end{document}